\documentclass[11pt,reqno]{amsart}

\usepackage{amssymb,mathtools}
\usepackage{url}
\usepackage{hyperref}
\usepackage[dvips]{color}
\usepackage{tikz-cd}
\usepackage[all]{xy}
\SelectTips{cm}{} 

\allowdisplaybreaks[4]

\usepackage{mathrsfs}
\let\mathcal\mathscr
\usepackage[mathscr]{eucal}

\let\phi=\varphi
\let\kappa=\varkappa
\let\epsilon=\varepsilon

\newcommand\restr[2]{{
  \left.\kern-\nulldelimiterspace 
  #1 
  \littletaller 
  \right|_{#2} 
  }}

\newcommand{\littletaller}{\mathchoice{\vphantom{\big|}}{}{}{}}

\newdir{ >}{{}*!/-5pt/\dir{>}}

\theoremstyle{theorem}
\newtheorem{classicalresult}{Classical Result}
\newtheorem{proposition}{Proposition}
\newtheorem{corollary}{Corollary}

\newtheorem{lemma}{Lemma}

\theoremstyle{definition}
\newtheorem{example}{Example}

\theoremstyle{remark}
\newtheorem{remark}{Remark}

\usepackage{mathrsfs}
\let\mathcal\mathscr
\usepackage[mathscr]{eucal}
\newcommand{\cprime}{\/{\mathsurround=0pt$'$}}

\author{Ji\v rina Jahnov\'a}
\address{Mathematical Institute, Silesian University in Opava, Na Rybn\'{\i}\v{c}ku 1, 746 01 Opava, Czech Republic}
\email{Jirina.Jahnova@math.slu.cz} 
\title{On the Characteristic Form of $\mathfrak{g}$-Valued Zero-Curvature Representations}

\begin{document}

\begin{abstract}
We study $\mathfrak{g}$-valued zero-curvature representations (ZCRs) for partial differential equations in two independent variables from the perspective of their extension to the entire infinite jet space, focusing on their characteristic elements.  Since conservation laws -- more precisely, conserved currents --  and their generating functions for a given equation are precisely the $\mathbb{R}$-valued ZCRs and their characteristic elements, a natural question arises: to what extent can results known for conservation laws be extended to general $\mathfrak{g}$-valued ZCRs.

For a fixed matrix Lie algebra $\mathfrak{g} \subset \mathfrak{gl}(n)$, we formulate ZCRs as equivalence classes of $\mathfrak{g}$-valued function pairs on the infinite jet space that satisfy the Maurer--Cartan condition. Our main result establishes that every such ZCR admits a characteristic representative -- i.e., a representative in which the Maurer--Cartan condition takes a characteristic form -- generalizing the characteristic form known for conservation laws. This form is preserved under gauge transformations and can thus be regarded as a kind of normal form for ZCRs. We derive a new necessary condition, independent of the Maurer--Cartan equation, that must be satisfied by any characteristic representative. This condition is trivial in the abelian case but nontrivial whenever $\mathfrak{g}$ is nonabelian. These findings not only confirm structural assumptions used in previous works but also suggest potential applications in the classification and computation of ZCRs.
\end{abstract}

\maketitle

\section*{Introduction} \label{sec:introduction}

Zero-curvature representations (ZCRs), first introduced in \cite{Zak}, provide a fundamental tool in the study of integrable partial differential equations (PDEs), encoding the integrability condition as the flatness of a connection. The existence of a ZCR for a given PDE is not only a key signature of its integrability but also a powerful source of structural insight: ZCRs containing a non-removable spectral parameter can often be used to construct B\"{a}cklund transformations, generate recursion operators, or yield infinite hierarchies of conservation laws.
ZCRs have been extensively studied in the literature, and their structural and classification properties continue to be an active area of research in the theory of integrable systems  (see, e.g., \cite{Bal15, Bal16, Igonin2019,IgoninManno2020,Ma02, Ma02-2,Sak04,Sa11}).

A key object associated with a ZCR is its characteristic element. Motivated by analogies with the $\mathcal{C}$-spectral sequence and geometric conservation laws, this notion was introduced in~\cite{Ma92} for $\mathfrak{g}$-valued ZCRs associated with a general partial differential equation, where $\mathfrak{g}$ is an arbitrary but fixed Lie algebra. In the abelian case $\mathfrak{g} = \mathbb{R}$, and in the case of two independent variables, this element coincides with the generating function of a conservation law. Independently, and without resorting to cohomological tools, this concept was also introduced by Sakovich for evolutionary equations~\cite{Sa95}.

In~\cite{Ma92} and further developed in~\cite{Ma97}, attention was primarily paid to a necessary condition that any characteristic element must satisfy. In the case of two independent variables, when \( \mathfrak{g} = \mathbb{R} \), this condition reduces to the well-known requirement that the generating function of a conservation law is a cosymmetry.  
Unlike the cosymmetry condition, however, in the general nonabelian case, the necessary condition for the characteristic element also involves the functions defining the ZCR itself. Hence, it cannot be solved independently. On the other hand, the fact that this condition imposes additional constraints on the ZCR, combined with the transformation behavior of the characteristic element under gauge transformations (namely, conjugation by elements of the associated Lie group), makes the characteristic element a valuable tool in the search for and classification of gauge equivalence classes of $\mathfrak{g}$-valued ZCRs for a given equation manifold $\mathcal{E}$ (see \cite{Ma92,Ma97,Se05,Se08}). For this reason, characteristic elements have typically been studied in the form of $N$-tuples of $\mathfrak{g}$-valued functions expressed in fixed internal coordinates on~$\mathcal{E}$.

In the theory of conservation laws, the key task is to determine whether a solution of the necessary condition -- that is, a cosymmetry -- is in fact the generating function of a genuine conservation law.  
The crucial observation is that, modulo trivial currents, the generating function can be extended from the equation manifold~$\mathcal{E}$ to the full jet space~$J^\infty(\pi)$ in such a way that the associated conservation law can be written in its so-called characteristic form. This extension process fundamentally relies on integration by parts (see, e.g., \cite{Boch,Ol93}) for functions on $J^\infty(\pi)$, and it cannot be carried out using only internal coordinates on the equation manifold.

The goal of this paper is to explore whether the above result (and its consequences) known for generating functions of conservation laws can be extended to the more general case of characteristic elements of $\mathfrak{g}$-valued ZCRs, where $\mathfrak{g}\subset \mathfrak{gl}(n)$ is \textit{any} fixed matrix Lie algebra.
To this end, we leave aside the coordinate-free cohomological  approach and instead study the characteristic element of  $\mathfrak{g}$-valued ZCRs from the jet-coordinates perspective. We focus on those $\mathfrak{g}$-valued ZCRs for a PDE~$\mathcal{E}$ in two independent variables that can be extended to the entire jet space~$J^\infty(\pi)$, still taking values in~$\mathfrak{g}$. We regard such ZCRs as equivalence classes of $\mathfrak{g}$-valued function pairs on~$J^\infty(\pi)$ that satisfy the Maurer-Cartan condition. This viewpoint is consistent with the approach taken e.g. by Balandin (see~\cite{Bal15,Bal16}), where ZCRs are considered as pairs of functions on the full jet space -- that is, as representatives of (some, and as we shall prove, all) elements of our equivalence classes.

Our principal result shows that for each such $\mathfrak{g}$-valued ZCR, one can always find a characteristic representative -- i.e., a representative in which the Maurer-Cartan condition, extended to the entire jet space, takes a specific form that, in the case $\mathfrak{g} = \mathbb{R}$, coincides precisely with the characteristic form of a conservation law (cf.~\cite{Ol93}). As a direct consequence of this result, we obtain a sufficient condition for a given function to be the characteristic element of a given ZCR.
It turns out that the characteristic form of the Maurer-Cartan condition is preserved under gauge transformations - the image of a characteristic representative under such a transformation is again a characteristic representative of the trasformed ZCR. This implies that the characteristic representative can be interpreted as a kind of a normal form of a ZCR. We also derive a new necessary condition that every characteristic representative must satisfy.
 This condition is independent of the Maurer--Cartan condition and appears to be nontrivial in the nonabelian setting. 
  
To the best of our knowledge, the obtained results have not appeared previously in the literature. We believe that this work could enhance the effectiveness of computational searches for $\mathfrak{g}$-valued ZCRs associated with a given equation, since existing approaches (\cite{Ma92, Ma97}) rely only on a necessary condition for the characteristic element, whereas our results provide also a sufficient condition. By analogy with the computational methods for conservation laws described in Table~5 of \cite{Wolf}, we may view the existing approaches to computing ZCRs as corresponding to steps IA and IB. By contrast, our results allow one to combine steps analogous to IB and IIA, as recommended by Wolf in the context of conservation laws (\cite{Wolf}).

We also believe that the results presented here can be used for classification purposes -- for example, to identify the most general form of a PDE which admits a given function pair as a representative of (an equivalence class of) a ZCR, since the results are formulated without the use of internal coordinates on any equation manifold.

The present paper is organized as follows. In Section 1, we briefly recall basic definitions and notations of the geometric theory of differential equations. We also provide a more detailed exposition of the theory of $\mathfrak{g}$-valued functions defined on the infinite jet space and on the diffiety, along with a description of the relevant algebraic structures these spaces carry.

Section 2 provides, for completeness, the definition of conservation laws for a diffiety in two independent variables, formulated as a specific equivalence class of pairs of functions on the jet space. The presentation given here deviates slightly from the standard terminology found in, for example, \cite{Ol93} or \cite{Wolf}, with the aim of clarifying the perspective that conservation laws are, in fact, a special case of zero-curvature representations. We also review several classical results on conservation laws, which will serve as a foundation for their generalizations to ZCRs in the following section.

Section 3 provides a more detailed exposition of the theory of $\mathfrak{g}$-valued zero-curvature representations (ZCRs) for a diffiety in two independent variables. We define ZCRs as specific equivalence classes of pairs of $\mathfrak{g}$-valued functions on the jet space. We also recall the necessary background on their associated characteristic elements. Furthermore, we revisit the notion of gauge transformations for ZCRs and examine their behaviour in the special case of conserved currents.

The main results of this paper are stated and proved in Section 4. Their presentation is structured to emphasize the analogy with the results on conservation laws formulated in Section 2.

\section{Preliminaries}\label{sec:Prel}
In this section, we recall the basic concepts of the geometric theory of differential equations; for further details, see, e.g., \cite{Boch, Kra17, Ol93}.
We also provide a more detailed exposition of functions defined on the entire jet space and on a diffiety, taking values in a general matrix Lie algebra~$\mathfrak{g}$.
For completeness and precision, we include several auxiliary results along with their proofs.
\label{sec:1}
\subsection{Jet space $J^{\infty}(\pi)$, the ring of differential functions on $J^{\infty}(\pi)$, the Lie algebra of $\mathfrak{g}$-valued functions on $J^{\infty}(\pi)$}

Below, we work in the infinite-dimensional jet space $J^{\infty}(\pi)$, where $\pi: \mathbb{R}^2 \times \mathbb{R}^m \to \mathbb{R}^2$ is the trivial bundle with coordinates $x$, $y$, and $u_I^k$ for $k = 1, \dots, m$. Here, $x$ and $y$ denote the independent variables, while $u_I^k$ represent the partial derivatives of the dependent variables with respect to the independent variables specified by the multiindex $I$ (the multiindex $I$ is symmetric and of arbitrary finite length $|I|$, consisting of symbols $x$ and $y$). The coordinate $u^k_{\emptyset}$ corresponding to the empty multiindex is typically denoted by $u^k$.

A \textit{differential function} on $J^{\infty}(\pi)$ is a smooth real-valued function $F$ that depends on the variables $x$, $y$, $u^k$, and on only finitely many jet variables $u_I^k$.
The ring of all differential functions on $J^{\infty}(\pi)$ is denoted by $\mathcal{F}(\pi)$. 

The infinite jet space $J^{\infty}(\pi)$ is equipped with two commuting vector fields known as the \textit{total derivative operators}:
\begin{equation*}
D_x = \frac{\partial}{\partial x} + \sum_{k,I} u_{I,x}^k \frac{\partial}{\partial u_I^k}, \qquad
D_y = \frac{\partial}{\partial y} + \sum_{k,I} u_{I,y}^k \frac{\partial}{\partial u_I^k}.
\end{equation*}
These vector fields generate the so-called \textit{Cartan distribution} $\mathcal{C}$ on $J^{\infty}(\pi)$. 

In what follows, $D_I$ denotes the composition of total derivative operators corresponding to the independent variables specified by the multiindex $I$.  A differential operator $$\Delta:(\mathcal{F}(\pi))^M\to(\mathcal{F}(\pi))^N,\ (f_1,\dots,f_M)\mapsto(g_1,\dots,g_N)$$ expressed in the terms of the total derivatives is called a \textit{$\mathcal{C}$-differential operator}; that is 
\begin{equation*}\Delta((f^1,\dots,f^M))=(g^1,\dots,g^N), \mbox{ where } g^i=\sum_{I,j}a^i_{I,j}D_I(f^j),\end{equation*} where $a^i_{I,j}\in\mathcal{F}(\pi)$.

The \textit{total divergence operator} is the $\mathcal{C}$-differential operator 
\begin{equation*}\mathrm{Div}:\ \mathcal{F}{(\pi)}\times\mathcal{F}(\pi)\to\mathcal{F}(\pi),\ (f_1,f_2)\mapsto D_xf_1+D_yf_2.\end{equation*}
The \textit{Euler operator} $\mathbf{E} = (\mathrm{E}_1, \dots, \mathrm{E}_m)$ is a differential operator  (not a $\mathcal{C}$-differential operator)
\begin{equation*}
\mathbf{E} : \mathcal{F}(\pi) \to \left( \mathcal{F}(\pi) \right)^m,
\end{equation*}
whose $k$-th component is given by the formula
\begin{equation*}
\mathrm{E}_k = \sum_I (-1)^{|I|} D_I \circ \frac{\partial}{\partial u_I^k}.
\end{equation*}
It is well known that 

\begin{equation}\label{eulerdiv}\mathrm{Im}\ \mathrm{Div}=\mathrm{ker}\ \mathbf{E}.\end{equation}\\

Let $\mathfrak{g} \subset \mathfrak{gl}(n)$ be a matrix Lie algebra. A \textit{$\mathfrak{g}$-valued (differential) function} on $J^{\infty}(\pi)$ is a smooth function $J^{\infty}(\pi) \to \mathfrak{g}$ that depends on only finitely many variables $u_I^k$. Such a function can be regarded as a matrix $M$ whose entries are differential functions, with the property that $M(\theta) \in \mathfrak{g}$ for all $\theta \in J^{\infty}(\pi)$. 
If we fix a basis  \( T_1, \dots, T_d \) in  \( \mathfrak{g} \), any $\mathfrak{g}$-valued function  \( M \) can be  written in the form  
\begin{equation}\label{decomp}
M = f^1 T_1 + \dots + f^d T_d,
\end{equation}
where \( f^1, \dots, f^d \in \mathcal{F}(\pi) \).

The set of all $\mathfrak{g}$-valued functions on $J^{\infty}(\pi)$ is denoted by $\mathcal{F}^{\mathfrak{g}}(\pi)$. It naturally inherits the Lie algebra structure from $\mathfrak{g}$, with the Lie bracket defined pointwise by
\begin{equation*}
[M, N](\theta) := [M(\theta), N(\theta)]
\end{equation*}
for all $\theta \in J^{\infty}(\pi)$.

Moreover, given a differential function $f \in \mathcal{F}(\pi)$ and a $\mathfrak{g}$-valued function $M \in \mathcal{F}^{\mathfrak{g}}(\pi)$, their product $f \cdot M \in \mathcal{F}^{\mathfrak{g}}(\pi)$ is defined pointwise as
\begin{equation*}
(f \cdot M)(\theta) := f(\theta) \cdot M(\theta),
\end{equation*}
where the dot on the right-hand side denotes scalar multiplication of the matrix $M(\theta)\in\mathfrak{g}$ by the scalar $f(\theta)\in\mathbb{R}$. In components, this means $(f \cdot M)_{ij} = f \cdot M_{ij}$, where the symbol $\cdot$ on the right-hand side denotes the multiplication in the ring $\mathcal{F}(\pi)$.

Since the total derivative operators act as differential operators on the ring $\mathcal{F}(\pi)$, they extend naturally to $\mathfrak{g}$-valued functions by differentiating entry-wise, that is, if $M\in\mathcal{F}^{\mathfrak{g}}(\pi)$, we define 
\begin{equation*}
(D_x M)_{ij} := D_x(M_{ij}), \qquad (D_y M)_{ij} := D_y(M_{ij}).
\end{equation*}
The resulting matrix is an element of $\mathcal{F}^{\mathfrak g}(\pi)$; indeed, if we apply the total derivative operator to the decomposition \eqref{decomp}, we obtain
\begin{equation*}D_xM=(D_xf^1)T_1+\dots+(D_xf^d)T_d,\end{equation*} since $T_1,\dots,T_d$ are matrices with constant entries. Hence, we have
$D_xM\in\mathcal{F}^{\mathfrak{g}}(\pi)$, and  $D_x$ and $D_y$ are thus linear maps $ \mathcal{F}^{\mathfrak{g}}(\pi)\to\mathcal{F}^{\mathfrak{g}}(\pi)$.

The following identities hold for any $f \in \mathcal{F}(\pi)$ and $M, A, B \in \mathcal{F}^{\mathfrak{g}}(\pi)$:
\begin{equation*}
D_x(f \cdot M) = D_x(f) \cdot M + f \cdot D_x(M), \ \ D_x[A, B] = [D_x A, B] + [A, D_x B].\end{equation*}
Analogous formulas hold for $D_y$.

Even though $\mathcal{F}^{\mathfrak{g}}(\pi)$ is not closed under the matrix multiplication, given $M,N\in\mathcal{F}^{\mathfrak{g}}(\pi)$, their product lies in $\mathcal{F}^{\mathfrak{gl}(n)}(\pi)$, and thus the operators $D_x,D_y$ can be applied to it. We have
\begin{equation*}\qquad D_x(M \cdot N) = D_x(M) \cdot N + M \cdot D_x(N).\end{equation*}

If we replace total differentiation everywhere above with partial differentiation, all results remain valid. In particular, $ \mathcal{F}^{\mathfrak{g}}(\pi)$ is closed under partial differentiation.

The observations above allow us to define the analogues of the total divergence operator and the Euler operator  for $\mathfrak{g}$-valued functions on $J^{\infty}(\pi)$. These are defined entrywise as follows:
\begin{equation*}
\mathrm{Div}^{\mathfrak{g}} \colon \mathcal{F}^{\mathfrak{g}}(\pi) \times \mathcal{F}^{\mathfrak{g}}(\pi) \to \mathcal{F}^{\mathfrak{g}}(\pi), \quad 
(\mathrm{Div}^{\mathfrak{g}}(M,N))_{ij} := \mathrm{Div}(M_{ij}, N_{ij}),
\end{equation*}

\begin{equation*}
\mathrm{E}^{\mathfrak{g}} = (\mathrm{E}^{\mathfrak{g}}_1, \dots, \mathrm{E}^{\mathfrak{g}}_m) 
\colon \mathcal{F}^{\mathfrak{g}}(\pi) \to \left(\mathcal{F}^{\mathfrak{g}}(\pi)\right)^m, \quad 
(\mathrm{E}^{\mathfrak{g}}_k(M))_{ij} := \mathrm{E}_k(M_{ij}),
\end{equation*}
where $M, N \in \mathcal{F}^{\mathfrak{g}}(\pi)$, $\mathrm{Div}$ denotes the standard total divergence operator  and $\mathrm{E}_k$ denotes the $k$-th component of the standard Euler operator.
The relationship between these operators is the same as in the case of $\mathbb{R}$-valued functions:
\begin{lemma}\label{eulerdivLie}
Let $\mathfrak{g}\subset\mathfrak{gl}$(n) be an arbitrary matrix Lie algebra. Then the following identity holds:
\begin{equation*}\mathrm{Im}\ \mathrm{Div}^{\mathfrak{g}}=\mathrm{ker}\ \mathrm{\mathbf{E}}^{\mathfrak{g}}.\end{equation*}
\end{lemma}
\begin{proof}
The inclusion $\operatorname{Im} \mathrm{Div}^{\mathfrak{g}} \subset \ker \mathbf{E}^{\mathfrak{g}}$ follows immediately from the definition of the operators under consideration and from~\eqref{eulerdiv}. From the same identity, we also deduce that if $\mathbf{E}^{\mathfrak{g}}(L) = 0$, then there exist $M, N \in \mathcal{F}^{\mathfrak{gl}(n)}(\pi)$ such that
$
D_x M + D_y N = L.
$
The key point is to show that in fact $M, N \in \mathcal{F}^{\mathfrak{g}}(\pi)\subset\mathcal{F}^{\mathfrak{gl}(n)}(\pi)$.\\ To this end, let us decompose $L$ with respect to a fixed basis $T_1, \dots, T_d$ of $\mathfrak{g}$:
\begin{equation*}
L = f^1 T_1 + \dots + f^d T_d,
\end{equation*}
and assume that $\mathbf{E}^{\mathfrak{g}}(L) = 0$. Then, for each $k = 1, \dots, m$,
\begin{equation*}
\mathrm{E}_k^{\mathfrak{g}}(L) = \mathrm{E}_k(f^1) T_1 + \dots + \mathrm{E}_k(f^d) T_d = 0.
\end{equation*}
Since $T_1, \dots, T_d$ are linearly independent over $\mathbb{R}$, it follows that $\mathrm{E}_k(f^i) = 0$ for all $i = 1, \dots, d$ and $k = 1, \dots, m$. In particular, each scalar function $f^i\in\mathcal{F}(\pi)$ satisfies $\mathrm{\mathbf{E}}(f^i) = 0$, and thus by~\eqref{eulerdiv}, we have
\begin{equation*}
f^i = D_x f^i_1 + D_y f^i_2
\end{equation*}
for some $f^i_1, f^i_2 \in \mathcal{F}(\pi)$. Putting
$
M := \sum_{i=1}^d f^i_1 T_i, \  N := \sum_{i=1}^d f^i_2 T_i,
$
we obtain 
\begin{equation*}L=D_xM+D_yN,\end{equation*} 
where
$M, N \in \mathcal{F}^{\mathfrak{g}}(\pi)$, as required.
\end{proof}

The differential Lie algebra $\mathcal{F}^{\mathfrak{g}}(\pi)$ can be equipped with more operators. In fact, any $\mathfrak{g}$-valued function $R \in \mathcal{F}^{\mathfrak{g}}(\pi)$ gives rise to two new operators $\mathcal{F}^{\mathfrak{g}}(\pi)\to\mathcal{F}^{\mathfrak{g}}(\pi)$:
\begin{equation}\label{TDhat}
\widehat{D}_x^R := D_x - \mathrm{ad}_R, \quad \widehat{D}_y^R := D_y - \mathrm{ad}_R,
\end{equation}
where $\mathrm{ad}_U(V) := [U, V]$ denotes the adjoint action. Both operators are derivations with respect to the Lie bracket.

However, we remark that, for two arbitrary functions $R$ and $S$, the operators $\widehat{D}_x^R$ and $\widehat{D}_y^S$ generally \textit{do not mutually commute} on $\mathcal{F}^{\mathfrak{g}}(\pi)$.\\

\subsection{Diffiety $\mathcal{E}$, the ring of smooth functions on $\mathcal{E}$, the Lie algebra of $\mathfrak{g}$-valued functions on~$\mathcal{E}$}
Let  
\begin{equation*}
F^l(x, y, \dots, u_I^k, \dots) = 0,\quad l = 1, \dots, N 
\end{equation*}
be a system of partial differential equations. Considering this system together with all its differential consequences  
$
D_I F^l = 0,\  l = 1, \dots, N,
$  
defines a submanifold $\mathcal{E} \subset J^{\infty}$:
\begin{equation*}
\mathcal{E} := \left\{ \theta \in J^{\infty}\ \middle|\ D_I(F^l)(\theta) = 0\ \text{for all } l, I \right\}.
\end{equation*}
The total derivative operators can be restricted to $\mathcal{E}$, and the manifold $\mathcal{E}$ equipped with the restriction of the Cartan distribution $\mathcal{C}|_{\mathcal{E}}$ constitutes an instance of a \emph{diffiety}. The total derivatives restricted to $\mathcal{E}$ will still be denoted by $D_x$ and $D_y$.

Throughout this paper, we assume that the manifold $\mathcal{E}$ satisfies the \textit{regularity condition} in the sense of \cite{Kra17}; that is, we assume that a function $f \in \mathcal{F}(\pi)$ vanishes on $\mathcal{E}$ if and only if there exists a $\mathcal{C}$-differential operator $\Delta$ such that $f = \Delta((F^1,\dots,F^N))$. This assumption enables us to define the ring of smooth functions on $\mathcal{E}$ as follows.

Let $\mathcal{I}(\mathcal{E})$ denote the differential ideal in $\mathcal{F}(\pi)$ generated by all total derivatives of the functions $F^l\in\mathcal{F}(\pi)$, $l=1,\dots,N$, i.e.,
\begin{equation*}
\mathcal{I}(\mathcal{E}) := \left\{ \sum_{I, l} f_l^I D_I(F^l)\ \bigm|\ f_l^I \in \mathcal{F}(\pi) \right\}.
\end{equation*}
The \textit{ring of smooth functions on $\mathcal{E}$} is defined as the factor ring
\begin{equation*}
\mathcal{F}(\mathcal{E}) := \mathcal{F}(\pi) / \mathcal{I}(\mathcal{E}) = \left\{ f + \mathcal{I}(\mathcal{E})\ \middle|\ f \in \mathcal{F}(\pi) \right\}.
\end{equation*}
The equivalence class $[f] \in \mathcal{F}(\pi) / \mathcal{I}(\mathcal{E})$ will be denoted by $f|_{\mathcal{E}}$, the zero element \( [0] \equiv \restr{0}{\mathcal{E}} \) of the factor ring will simply be denoted by \( 0 \). 
We thus have $f|_{\mathcal{E}} = g|_{\mathcal{E}}$ if and only if there exist differential functions $h_l^I \in \mathcal{F}(\pi)$ such that
\begin{equation*}
f = g + \sum_{I, l} h_l^I D_I(F^l).
\end{equation*}
In particular,  \( \restr{f}{\mathcal{E}} = 0 \) if and only if \( f \in \mathcal{I}(\mathcal{E}) \).

Conversely, given a function $f|_{\mathcal{E}} \in \mathcal{F}(\mathcal{E})$, any $g \in \mathcal{F}(\pi)$ such that $g|_{\mathcal{E}} = f|_{\mathcal{E}}$ is called an \emph{extension of $f|_{\mathcal{E}}$ to $J^{\infty}(\pi)$}.

The action of the total derivative operators on $\mathcal{F}(\mathcal{E})$ is given by
\begin{equation*}
D_x(f|_{\mathcal{E}}) := \restr{(D_x f)}{\mathcal{E}},\quad D_y(f|_{\mathcal{E}}) := \restr{(D_y f)}{\mathcal{E}}
\end{equation*}

Let $\mathfrak{g} \subset \mathfrak{gl}(n)$ be a matrix Lie algebra, and let $\mathcal{F}^{\mathfrak{g}}(\pi)$ denote the corresponding Lie algebra of $\mathfrak{g}$-valued functions on $J^{\infty}(\pi)$. 
The regularity condition imposed on $\mathcal{E}$ implies that an element $M\in\mathcal{F}^{\mathfrak g}(\pi)$ vanishes on $\mathcal{E}$ if and only if each of its entries lies in the differential ideal $\mathcal{I}(\mathcal{E})$. This leads to the following definition:

Define the ideal
\begin{equation*}
\mathcal{I}^{\mathfrak{g}}(\mathcal{E}) = \left\{ M \in \mathcal{F}^{\mathfrak{g}}(\pi)\ \middle|\ M_{ij} \in \mathcal{I}(\mathcal{E}) \right\}
\end{equation*}
as the set of all $\mathfrak{g}$-valued functions whose entries belong to the differential ideal $\mathcal{I}(\mathcal{E})$.

The Lie algebra of \textit{$\mathfrak{g}$-valued functions on $\mathcal{E}$} is then defined as the factor Lie algebra
\begin{equation*}
\mathcal{F}^{\mathfrak{g}}(\mathcal{E}) := \mathcal{F}^{\mathfrak{g}}(\pi)\big/\mathcal{I}^{\mathfrak{g}}(\mathcal{E}) = \left\{ M + \mathcal{I}^{\mathfrak{g}}(\mathcal{E})\ \middle|\ M \in \mathcal{F}^{\mathfrak{g}}(\pi) \right\}.
\end{equation*}
The equivalence class $[M] \in \mathcal{F}^{\mathfrak{g}}(\pi)/\mathcal{I}^{\mathfrak{g}}(\mathcal{E})$ is denoted by $\restr{M}{\mathcal{E}}$, the zero element \( [0] \equiv \restr{0}{\mathcal{E}} \) of the factor Lie algebra will simply be denoted by \( 0 \). We thus have $\restr{M}{\mathcal{E}} = \restr{N}{\mathcal{E}}$ if and only if, for all $i,j = 1, \dots, n$, their entries satisfy $\restr{M_{ij}}{\mathcal{E}} = \restr{N_{ij}}{\mathcal{E}}$, that is, if and only if
\begin{equation*}
M_{ij} = N_{ij} + \sum_{I,l} h_{ijl}^I D_I(F^l),
\end{equation*}
for some differential functions $h_{ijl}^I \in \mathcal{F}(\pi)$. In particular, \( \restr{M}{\mathcal{E}} = 0 \) if and only if  \( M \in \mathcal{I}^{\mathfrak{g}}(\mathcal{E}) \).

Having fixed a basis $T_1,\dots, T_d$ in $\mathfrak{g}$, one may ask how to determine -- based on its decomposition~\eqref{decomp} -- whether a given $\mathfrak{g}$-valued function $M\in\mathcal{F}^{\mathfrak{g}}(\pi)$ belongs to $\mathcal{I}^{\mathfrak g}(\mathcal{E})$: the answer to this question is given in the following lemma:
\begin{lemma}\label{decomp_lm}
Let $T_1,\dots, T_d$ be a basis of $\mathfrak{g}$ and
let $M\in\mathcal{F}^{\mathfrak g}(\pi)$ be a $\mathfrak{g}$-valued function on $J^{\infty}(\pi)$. Let $$M=f^1T_1+\dots f^dT_d, \mbox{ where } f^i\in\mathcal{F}(\pi)$$ be its decomposition \eqref{decomp} with respect to this basis.
Then, $M\in\mathcal{I}^{\mathfrak{g}}(\mathcal{E})$ if and only if all the functions $f^1,\dots,f^d$ belong to the ideal $\mathcal{I}(\mathcal{E})$. \end{lemma}
\begin{proof}
By the very definition of the ideal $\mathcal{I}^{\mathfrak{g}}(\mathcal{E})$ and the regularity assumption, the function $M$ belongs to it if and only if $M(\theta)=0$ for all $\theta\in\mathcal{E}$. Since 
\begin{equation*}M(\theta)=f^1(\theta)T_1+\dots+f^d(\theta)T_d\end{equation*}
 and the elements $T_1,\dots,T_d$ are linearly independent, it follows that $M(\theta)=0$ for all $\theta\in\mathcal{E}$ if and only if $f^1(\theta)=f^2(\theta)=\dots=f^d(\theta)=0$ for all $\theta\in\mathcal{E}$.\\
By the regularity assumption it is equivalent to $f^1,\dots,f^d\in\mathcal{I}(\mathcal{E})$.\end{proof}
Consequently, if $M=f^1T_1+\dots f^dT_d$ and $N=g^1T_1+\dots g^dT_d$ are the decompositions of two $\mathfrak{g}$-valued functions, then $\restr{M}{\mathcal{E}}=\restr{N}{\mathcal{E}}$ if and only if $\restr{f^i}{\mathcal{E}}=\restr{g^i}{\mathcal{E}}$ for all $i=1,\dots,d$. Thus, each equivalence class $\restr{M}{\mathcal{E}}$ has a unique decomposition of the form
\begin{equation*}\restr{M}{\mathcal{E}}=(\restr{f^1}{\mathcal{E}})T_1+\dots +(\restr{f^d}{\mathcal{E}})T_d\end{equation*}
 with respect to the fixed basis $T_1,\dots,T_d$.

The assertion stated in the following lemma will be crucial for the definition of the characteristic element of a ZCR:
\begin{lemma}\label{decomp_lm2}
A $\mathfrak{g}$-valued function $M\in\mathcal{F}^{\mathfrak g}(\pi)$ belongs to the ideal $\mathcal{I}^{\mathfrak{g}}(\mathcal{E})$ if and only if there exist $\mathfrak{g}$-valued functions $C_l^I\in\mathcal{F}^{\mathfrak g}(\pi)$ such that
\begin{equation*}M=\sum_{I,l}D_I(F^l)C^I_l.\end{equation*}
\end{lemma}
\begin{proof}
Let $T_1,\dots,T_d$ be an arbitrary basis in $\mathfrak{g}$. Consider the decomposition \eqref{decomp} $M=f^1T_1+\dots+f^dT_d$ of $M$ with respect to this basis, where  $f^i\in\mathcal{F}(\pi)$, $i=1,\dots,d$.
 As $M\in\mathcal{I}^{\mathfrak{g}}(\mathcal{E})$, Lemma \ref{decomp_lm} implies that each  $f^i$ belongs to the ideal $\mathcal{I}(\mathcal{E})$.\\
 By definition of $\mathcal{I}(\mathcal{E})$, this means that for each $i=1,\dots,d$, there exist functions $g_{l}^{iI} \in \mathcal{F}(\pi)$ such that $f^i=\sum_{I,l}g_{l}^{iI}D_I(F^l)$. Substituting into the expression for $M$, we obtain:
 \begin{equation*}M=\sum_{i=1}^d\left(\sum_{I,l}g_{l}^{iI}D_I(F^l)\right)T_i=\sum_{I,l}D_I(F^l)\left(\sum_{i=1}^dg_{l}^{iI}T_i\right).\end{equation*}
Setting $C_l^I:=\sum_{i=1}^dg_{l}^{iI}T_i\in\mathcal{F}^{\mathfrak g}(\pi)$, we obtain the desired decomposition.\\
The opposite direction is obvious.
\end{proof}

Similarly to the case of $\mathcal{F}(\mathcal{E})$, the total derivative operators $D_x$ and $D_y$ are well defined on $\mathcal{F}^{\mathfrak{g}}(\mathcal{E})$:
\begin{equation*}D_x(\restr{M}{\mathcal{E}}):=\restr{(D_xM)}{\mathcal{E}},\ D_y(\restr{M}{\mathcal{E}}):=\restr{(D_yM)}{\mathcal{E}}\end{equation*} for any $\restr{M}{\mathcal{E}}\in\mathcal{F}^{\mathfrak{g}}(\mathcal{E})$.

 The operators $\widehat{D}_x^R$ and $\widehat{D}_y^R$, associated with an arbitrary function $R \in \mathcal{F}^{\mathfrak{g}}(\pi)$, can also be consistently restricted to $\mathcal{F}^{\mathfrak{g}}(\mathcal{E})$ as follows:
\begin{equation*}
\widehat{D}_x^R(\restr{M}{\mathcal{E}}) := \restr{(D_x M - \mathrm{ad}_R(M))}{\mathcal{E}}, \quad
\widehat{D}_y^R(\restr{M}{\mathcal{E}}) := \restr{(D_y M - \mathrm{ad}_R(M))}{\mathcal{E}}.
\end{equation*}
Finally, the operators $\widehat{D}_x^{\widetilde{R}}$ and $\widehat{D}_x^R$ (and similarly $\widehat{D}_y^{\widetilde{R}}$ and $\widehat{D}_y^R$) coincide on $\mathcal{F}^{\mathfrak{g}}(\mathcal{E})$ whenever $\restr{R}{\mathcal{E}} = \restr{\widetilde{R}}{\mathcal{E}}$, that is, for any $\restr{M}{\mathcal{E}} \in \mathcal{F}^{\mathfrak{g}}(\mathcal{E})$, we have:
\begin{equation*}
\widehat{D}_x^{\widetilde{R}}(\restr{M}{\mathcal{E}}) = \widehat{D}_x^R(\restr{M}{\mathcal{E}}), \quad
\widehat{D}_y^{\widetilde{R}}(\restr{M}{\mathcal{E}}) = \widehat{D}_y^R(\restr{M}{\mathcal{E}}).
\end{equation*}
Thus, any function $\restr{R}{\mathcal{E}} \in \mathcal{F}^{\mathfrak{g}}(\mathcal{E})$ defines the pair $\widehat{D}_x^{\restr{R}{\mathcal{E}}}, \widehat{D}_y^{\restr{R}{\mathcal{E}}}$ of well-defined operators on $\mathcal{F}^{\mathfrak{g}}(\mathcal{E})$.

Without any confusion, we will denote all the operators derived above just as $\widehat{D}_x^R$ and $\widehat{D}_y^R$. Their domain of definition is then clear from the context. 
\section{Conservation laws: classical results}\label{sec:CL}
This section provides the basic definitions of conservation laws and their generating functions for partial differential equations in two independent variables.
 The exposition slightly differs from the standard one, with the aim of making it more transparent that conservation laws can be viewed as special cases of zero-curvature representations (ZCRs), and that gauge-equivalent $\mathbb{R}$-valued ZCRs correspond precisely to equivalent conservation laws. 
We also recall classical results concerning conservation laws and their generating functions. These results will serve as reference points in the subsequent sections, where we attempt to generalize them to the setting of $\mathfrak{g}$-valued zero-curvature representations.
For a general theory of conservation laws, see e.g.~\cite{Boch, Kra17, Ol93, Wolf}.

Let $\mathcal{E}$ be a diffiety in two independent variables.
A \textit{conserved current} for $\mathcal{E}$ is understood to be the two-tuple $\restr{P}{\mathcal{E}}=(\restr{P_1}{\mathcal{E}},\restr{P_2}{\mathcal{E}})\in\mathcal{F}(\mathcal{E})\times\mathcal{F}(\mathcal{E})$ of smooth functions on $\mathcal{E}$ such that one (and hence every) pair of its representatives  satisfies the \textit{conservation law condition}
\begin{eqnarray}\label{CL}0=\restr{(D_xP_1+D_yP_2)}{\mathcal{E}}=\restr{\mathrm{Div} P}{\mathcal{E}}.\end{eqnarray} 
\begin{remark}
In the terminology of \cite{Ol93}, the equality \eqref{CL} is called the conservation law for $\mathcal{E}$. Two conservation laws
\begin{equation*}\restr{(D_xP_1+D_yP_2)}{\mathcal{E}}=0 \mbox{ \ and\  }\restr{(D_x\widetilde P_1+D_y\widetilde P_2)}{\mathcal{E}}=0\end{equation*}
where the pairs   
$
P = (P_1, P_2)$ and $\widetilde{P} = (\widetilde{P}_1, \widetilde{P}_2)$
represent the same conserved current, i.e.,  
$
\left. P \right|_{\mathcal{E}} = \left. \widetilde{P} \right|_{\mathcal{E}},
$
are then said to be \emph{equivalent under the equivalence of the first kind}.

\end{remark}
Contrary to the commonly used terminology, we will understand the term \emph{conservation law} (CL) below to refer to the entire equivalence class of conservation laws under the equivalence of the first kind as described in the preceding remark. This yields a one-to-one correspondence between conservation laws and conserved currents, allowing us to use the terms \emph{conservation law} and \emph{conserved current} interchangeably without any ambiguity.\\

In the general theory, a so-called
generating function  is assigned to each conservation law. In coordinates, it can be defined as follows: 

Let $(\restr{P_1}{\mathcal{E}}, \restr{P_2}{\mathcal{E}})$ be a conserved current (=conservation law) for $\mathcal{E}$. Consider an arbitrary extension  $(P_1,P_2)$ of it to $J^{\infty}(\pi)$ and rewrite the corresponding conservation law condition \eqref{CL} in the form 
\begin{equation}\label{CL2}\mathrm{Div} P=D_xP_1+D_yP_2= \sum_{I,l}(D_IF^l)C_l^I\ \ \mbox{ on } J^{\infty}(\pi),\end{equation}
where the coefficients $C_l^I\in\mathcal{F}(\pi)$ are some differential functions (it is possible; see the definition of the zero function on $\mathcal{E}$ in the first section).
Then the \textit{generating function} $\psi=(\psi_1,\dots,\psi_N)\in(\mathcal{F}(\mathcal{E}))^N$ is the $N$-tuple of functions on $\mathcal{E}$, where
\begin{equation}\psi_l:=\sum_{I}(-1)^{|I|}\restr{D_I(C_l^I)}{\mathcal{E}}\label{gen-funct}.\end{equation}
It is well known that each generating function $\psi$ is always a \textit{cosymmetry} for $\mathcal{E}$, that is, it satisfies the condition
\begin{equation}\label{cosym}
\sum_{I,l}(-1)^{|I|}\restr{D_I\left(\frac{\partial F^l}{\partial u^k_I}\psi_l\right)}{\mathcal{E}}=0 \ \mbox{ for all } k=1,\dots,m.
\end{equation}
However, not every cosymmetry, may be the generating function of a conservation law. Determining whether a cosymmetry genuinely generates a conservation law, based solely on the definition, is a highly nontrivial task. It ultimately amounts to verifying the existence of all functions involved in the definition \eqref{gen-funct} of the generating function -- a challenge further complicated by the fact that the number of such functions is, in general, not known in advance.
This task can be significantly simplified by making use of the fact that  every conservation law  can be written in the so-called \textit{characteristic form} (see \cite{Ol93}), which is obtained  by simply integrating by parts the equality \eqref{CL2}. We now state this fundamental result explicitly in a form adapted to our setting, cf. \cite{Boch,Kra17,Ol93}:

\begin{classicalresult}\label{prop1}
Let $(\restr{P_1}{\mathcal{E}}, \restr{P_2}{\mathcal{E}})$ be a conserved current for $\mathcal{E}$. Then there exists
its extension $(\widetilde{P}_1, \widetilde{P}_2)$ to $J^{\infty}(\pi)$
  and an $N$-tuple $Q = (Q_1, \dots, Q_N) \in \mathcal{F}^N$ such that  the identity
\begin{equation}\label{CL_char_form}D_x(\widetilde P_1)+D_y(\widetilde P_2)=\sum_{l=1}^NF^lQ_l,\end{equation}
is satisfied on $J^{\infty}(\pi)$.\\
The restriction $\restr{Q}{\mathcal{E}}=(\restr{Q_1}{\mathcal{E}},\dots,\restr{Q_N}{\mathcal{E}})$ is precisely the generating function of the conservation law corresponding to the conserved current $(\restr{P_1}{\mathcal{E}},\restr{P_2}{\mathcal{E}})=(\restr{\widetilde P_1}{\mathcal{E}},\restr{\widetilde P_2}{\mathcal{E}})$.
\end{classicalresult}

The pair $(\widetilde P_1,\widetilde P_2)\in\mathcal{F}(\pi)\times\mathcal{F}(\pi)$ of functions realizing the characteristic form of the conservation laws is referred to as \textit{a characteristic representative of the conserved current} $(\restr{P_1}{\mathcal{E}},\restr{P_2}{\mathcal{E}})$ while the $N$-tuple $Q=(Q_1,\dots,Q_N)\in\mathcal{F}(\pi)^{N}$ is called  the \textit{characteristic} corresponding to the characteristic representative $(\widetilde P_1,\widetilde P_2)$.

Thus, the generating function of any conservation law $(\restr{P_1}{\mathcal{E}},\restr{P_2}{\mathcal{E}})$ arises as the restriction of a characteristic associated with some characteristic representative of the conserved current $(\restr{P_1}{\mathcal{E}},\restr{P_2}{\mathcal{E}})$:

\begin{classicalresult}\label{prop3}
The $N$-tuple of functions \( \psi = (\psi_1, \dots, \psi_N) \in(\mathcal{F}(\mathcal{E}) )^N\) is the generating function of the conservation law $(\restr{P_1}{\mathcal{E}},\restr{P_2}{\mathcal{E}})$ if and only if there exists an extension 
\(( \tilde P_1, \tilde P_2) \) of the conserved current to $J^{\infty}(\pi)$ and the extension \( Q=(Q_1, \dots, Q_N ) \) of $\psi$ to $J^{\infty}(\pi)$ such that  condition~\eqref{CL_char_form} is satisfied on $J^{\infty}(\pi)$.
\end{classicalresult}

 This result significantly simplifies the aforementioned problem of determining whether a given cosymmetry $\psi$ is the generating function of some conservation law.  It is sufficient to determine whether $\psi$ admits an extension $Q$ to $J^{\infty}(\pi)$ such that equation~\eqref{CL_char_form} can be satisfied for some (a priori unknown) functions $P_1,P_2$. The exactness of the global variational complex of the bundle $\pi$ -- specifically, the identity $ \mathrm{ker}\ \mathbf{E}= \mathrm{im}\ \mathrm{Div}$ on $J^{\infty}(\pi)$ -- makes it possible to answer this question without explicitly constructing the unknown functions (for more details, see \cite{Boch,Ol93}):

\begin{classicalresult}\label{prop4}
The $N$-tuple of functions \( \psi = (\psi_1, \dots, \psi_N) \in(\mathcal{F}(\mathcal{E}) )^N\) is the generating function of a conservation law if and only if it is possible to find its extension
 \( Q=(Q_1, \dots, Q_N) \in \mathcal{F}(\pi)^N \) on $J^{\infty}(\pi)$
 such that the condition
 \begin{equation}\label{char_suf}
 0=\sum_{l=1}^N\sum_I(-1)^{|I|}D_I\left(\frac{\partial F^l}{\partial u_I^k}Q_l\right)+\sum_{l=1}^N\sum_I(-1)^{|I|}D_I\left(\frac{\partial Q^l}{\partial u_I^k}F_l\right)
 \end{equation}
  is satisfied identically on $J^{\infty}(\pi)$ for all $k=1,\dots,m$.
\end{classicalresult}

Conservation laws are usually considered modulo trivial ones. To introduce the corresponding equivalence relation, we introduce the following group action.

Let $R\in\mathcal{F}(\pi)$ be a differential function on $J^{\infty}(\pi)$ and $\restr{P}{\mathcal{E}}=(\restr{P_1}{\mathcal{E}},\restr{P_2}{\mathcal{E}})$ be a conserved current for $\mathcal{E}$. Then the two-tuple 
\begin{equation*}t_R(\restr{P}{\mathcal{E}})\equiv\restr{P^R}{\mathcal{E}}=(\restr{P_1^R}{\mathcal{E}},\restr{P_2^R}{\mathcal{E}}),\end{equation*} 
where\begin{equation*}P^R_1=P_1-D_yR \mbox{ \ and \ } P^R=P_2+D_xR,\end{equation*} is a conserved current again. 

It can be easily seen that the prescription $t:(R,\restr{P}{\mathcal{E}})\mapsto \restr{P^R}{\mathcal{E}}$ provides us with an action of the additive group $\mathcal{F}(\pi)$ on the set of all conserved currents for $\mathcal{E}$. The conserved currents lying in the same orbit are called \textit{equivalent conserved currents} and the conservation laws corresponding to two equivalent conserved currents are then said to be $\textit{equivalent conservation laws}$. 
The conserved current (or conservation law) is called \textit{trivial}, if it is \textit{equivalent} to the conserved current $(\restr{0}{\mathcal{E}},\restr{0}{\mathcal{E}})$. 

Thus, the two-tuple $(P_1,P_2)\in\mathcal{F}(\pi)\times\mathcal{F}(\pi)$ of differential functions that satisfy the condition \eqref{CL} represents the trivial conserved current (conservation law) if and only if there exists a differential function $R\in\mathcal{F}(\pi)$ and two differential functions $X,Y\in\mathcal{I}(\mathcal{E})$ such that \begin{equation*}(P_1,P_2)=(-D_yR+X,D_xR+Y).\end{equation*}
\begin{remark}
This aligns exactly with the equivalence of conservation laws that encompasses both the \textit{first-kind} and \textit{second-kind} equivalences as is defined in \cite{Ol93, Wolf}.
\end{remark}

The following proposition describes the behaviour of a characteristic representative of a conserved current and its corresponding characteristic under the action of the previously introduced transformation $t_R$ applied to a conserved current. The statement is a reformulation of well-known results presented e.g. in \cite {Boch,Kra17,Ol93}:
\begin{classicalresult}\label{prop_CL_char}
Let the pair of differential functions $(P_1, P_2)$ be a characteristic representative of the conserved current $(\restr{P_1}{\mathcal{E}}, \restr{P_2}{\mathcal{E}})$, and let $R \in \mathcal{F}(\pi)$ be a differential function. Then the pair $P^R = (P_1-D_y R,\, P_2+D_x R)$ is a characteristic representative of the transformed conserved current
\begin{equation*}(\restr{P^R_1}{\mathcal{E}}, \restr{P^R_2}{\mathcal{E}}) = (\restr{P_1 - D_y R}{\mathcal{E}}, \restr{P_2 + D_x R}{\mathcal{E}}),\end{equation*}
with the corresponding characteristic remaining unchanged.
\end{classicalresult}

Thus, the generating function is an object that is well-defined for the whole equivalence class of conserved currents (see also \cite{Boch}). Moreover, it is well known that the generating function determines the equivalence class of conserved currents uniquely (\cite{Boch,Ol93})
\begin{remark}
In the general theory (in the two-dimensional base setting), conservation laws are understood to be closed horizontal differential $1$-forms on the diffiety $\mathcal{E}$ (see \cite{Boch, Kra17}). The conserved currents described above then appear as coordinate coefficients of such forms, and the conservation law condition \eqref{CL} expresses precisely the condition that the form is closed.
The equivalence classes of conserved currents introduced above correspond exactly to the first cohomology classes of the horizontal de Rham complex on $\mathcal{E}$. Trivial conserved currents thus correspond to exact horizontal differential $1$-forms on $\mathcal{E}$.

The generating function of a conservation law is, in general, defined as the image of the first cohomology group of the de Rham complex on $\mathcal{E}$ under a certain map arising in the first `sheet' of the $\mathcal{C}$-spectral sequence (see \cite{Boch}). However, this map acts on entire cohomology classes, and hence it is well defined on any individual representative of such a class. It follows that one can correctly define the generating function for each conserved current individually, as we have done above.
\end{remark}

Finally, we emphasize that in the standard theory of conservation laws, a conservation law is typically regarded as an entire equivalence class of conserved currents, as described above. However, in view of our aim to draw an analogy between the theory of conservation laws and zero-curvature representations, we shall consistently \textit{distinguish between individual conserved currents as defined above and their corresponding equivalence classes.}

\section{The zero curvature representation}\label{sec:ZCR}
In this section, we introduce the basic definitions and results concerning $\mathfrak{g}$-valued zero-curvature representations, their characteristic elements, and gauge transformations. Our goal is to lay a solid foundation for generalizing classical results known for conservation laws, which, however, require extending the considered $\mathfrak{g}$-valued functions from the diffiety $\mathcal{E}$ to the entire infinite jet space $J^{\infty}(\pi)$. We thus follow, more or less, the exposition of ZCRs presented in~\cite{Ma92}, however, everything is reformulated in terms of jet coordinates. For the sake of completeness and rigor, we also provide proofs of several statements that are not always immediately obvious.\\

Let $\mathcal{E}$ be a diffiety in two independent variables, 
let $\mathfrak{g}\subset\mathfrak{gl}(n)$ be an arbitrary matrix Lie algebra. A \textit{$\mathfrak{g}$-valued zero-curvature representation} (ZCR) for $\mathcal{E}$ is  defined to be the pair $\restr{(A,B)}{\mathcal{E}}:=(\restr{A}{\mathcal{E}},\restr{B}{\mathcal{E}})\in\mathcal{F}^{\mathfrak{g}}(\mathcal{E})\times\mathcal{F}^{\mathfrak{g}}(\mathcal{E})$ of the  $\mathfrak{g}$-valued functions on $\mathcal{E}$ that satisfy the \textit{Maurer--Cartan condition} 
 \begin{equation}\label{ZCR}
\restr{(D_yA-D_xB+[A,B])}{\mathcal{E}}=0.
\end{equation}
\begin{lemma}
The above definition is well-defined; that is, the condition~\eqref{ZCR} depends only on the equivalence classes $\restr{A}{\mathcal E}$ and $\restr{B}{\mathcal E}$, not on the particular representatives $A,B\in \mathcal{F}^{\mathfrak{g}}(\pi)$.
\end{lemma}
\begin{proof}Let the pair $(A,B)$  of $\mathfrak{g}$-valued functions satisfy the condition \eqref{ZCR}. Let $(\widetilde A,\widetilde B)$ be another pair of  $\mathfrak{g}$-valued functions equivalent to $(A,B)$, that is $\restr{A}{\mathcal{E}}=\restr{\widetilde A}{\mathcal{E}}$ and $\restr{B}{\mathcal{E}}=\restr{\widetilde B}{\mathcal{E}}$. We claim that the pair $(\tilde A,\tilde B)$ satisfies the condition \eqref{ZCR} as well. \\
Indeed, according to the definition, 
we have  $\widetilde A=A+X$ and $\widetilde B=B+Y$ where $X,\ Y$ are some $\mathfrak{g}$-valued functions from $\mathcal{I}^{\mathfrak{g}}(\mathcal{E})$.
Then \begin{align*}(D_y\tilde A-D_x\tilde B+[\tilde A,\tilde B])&=(D_y(A+X)-D_x(B+Y)+[A+X,B+Y])\\
&=D_yA-D_xB+[A,B]+(D_yX-D_xY+[A,Y]+[X,B]+[X,Y])\\
&=D_yA-D_xB+[A,B]+Z.\end{align*}
The term $Z=(D_yX-D_xY+[A,Y]+[X,B]+[X,Y])$ lies in $\mathcal{I}^{\mathfrak{g}}(\mathcal{E})$ since the operators $D_x$ and $D_y$ obviously map matrices from $\mathcal{I}^{\mathfrak{g}}(\mathcal{E})$ to $\mathcal{I}^{\mathfrak{g}}(\mathcal{E})$, and $\mathcal{I}^{\mathfrak{g}}(\mathcal{E})$ is an ideal in $\mathcal{F}^{\mathfrak{g}}(\pi)$.  Thus, $\restr{(D_y\tilde A-D_x\tilde B+[\tilde A,\tilde B])}{\mathcal{E}}=\restr{(D_yA-D_xB+[A,B])}{\mathcal{E}}=0$, that was to be verified.

\end{proof}

\begin{remark}
In the case where $\mathfrak{g}=\mathbb{R}$ and $(\restr{A}{\mathcal{E}},\restr{B}{\mathcal{E}})$ is an $\mathbb{R}$-valued ZCR, the pair $(-\restr{B}{\mathcal{E}},\restr{A}{\mathcal{E}}) $ is a conserved current. Conversely, whenever $(\restr{P_1}{\mathcal{E}},\restr{P_2}{\mathcal{E}})$ is a conserved current for $\mathcal{E}$, the pair $(-\restr{P_2}{\mathcal{E}},\restr{P_1}{\mathcal{E}})$ defines an $\mathbb{R}$-valued ZCR.  Based on this one-to-one correspondence, conserved currents can be regarded as a special case of zero-curvature representations.
\end{remark}

Similarly to the case of conservation laws, in the context of zero-curvature representations we will also need to rely on the regularity assumption for our equation in order to extend the identity~\eqref{ZCR} to the entire jet space. According to Lemma \ref{decomp_lm2} the identity \eqref{ZCR} is equivalent to
\begin{equation}\label{ZCR2}
D_yA-D_xB+[A,B]=\sum_{I,l}D_I(F^l)\cdot C_l^I,
\end{equation} 
where  \( C_l^I \in \mathcal{F}^{\mathfrak{g}}(\pi) \), and the identity holds on $J^{\infty}(\pi)$.

Let \( (\restr{A}{\mathcal{E}}, \restr{B}{\mathcal{E}}) \) be a fixed $\mathfrak{g}$-valued zero-curvature representation for $\mathcal{E}$, and let \( (A, B) \) be an arbitrary representative of this ZCR.  
From now on, we denote the associated operators \( \widehat{D}_x^{A}\) ($\equiv\widehat{D}_x^{\restr{A}{\mathcal{E}}}$) and \( \widehat{D}_y^B \) ($\equiv\widehat{D}_y^{\restr{B}{\mathcal{E}}}$)  on \( \mathcal{F}^{\mathfrak{g}}(\mathcal{E}) \) as follows:
\begin{equation*}
\widehat{D}_x := \widehat{D}_x^A \quad \text{and} \quad \widehat{D}_y := \widehat{D}_y^B.
\end{equation*}
It is straightforward to verify that these operators commute on \( \mathcal{F}^{\mathfrak{g}}(\mathcal{E}) \).  
Therefore, when working on \( \mathcal{F}^{\mathfrak{g}}(\mathcal{E}) \), we introduce a simplified notation for their compositions:  
given a multiindex \( I = (a,b) \), we define (see \cite{Ma97})
\begin{equation}
\widehat{D}_I = \underbrace{\widehat{D}_{x}\circ\dots\circ\widehat{D}_{x}}_{a \text{ times}}\circ\underbrace{\widehat{D}_{y}\circ\dots\circ\widehat{D}_{y}}_{b \text{ times}}\
:= (\widehat{D}_x)^a\, (\widehat{D}_y)^b,
\end{equation}
where the operator \( \widehat{D}_x \) is applied \( a \) times, and \( \widehat{D}_y \) is applied \( b \) times.

In general, given a \( \mathfrak{g} \)-valued ZCR \( \restr{(A,B)}{\mathcal{E}} \) for \( \mathcal{E} \), one can associate to it the so-called linear \textit{gauge complexes}.  
An important object that arises -- up to isomorphism of complexes -- as a particular representative of the first cohomology class of the first linear gauge complex is referred to as the \textit{characteristic element} of the ZCR (for more details, see \cite{Ma92}). Its general definition is rather involved, and we shall not present it here.  

However, as stated in \cite{Ma92}, when \( \restr{(A,B)}{\mathcal{E}} \) is a \( \mathfrak{g} \)-valued ZCR whose extension to $J^{\infty}(\pi)$ admits a decomposition of the form ~\eqref{ZCR2} (according to Lemma \ref{decomp_lm2} such decomposition always exists), the corresponding characteristic element \( \chi \) is given explicitly by the \( N \)-tuple
\begin{equation*}
\chi = \left( \restr{K_1}{\mathcal{E}}, \dots, \restr{K_N}{\mathcal{E}} \right) \in \left( \mathcal{F}^{\mathfrak{g}}(\mathcal{E}) \right)^N,
\end{equation*}
where each component $K_l$ is defined by
\begin{equation} \label{char-el}
K_l = \sum_I (-1)^{|I|} \widehat{D}_I C_l^I.
\end{equation}

It was proved in \cite{Ma92} that, given a \( \mathfrak{g} \)-valued ZCR, the corresponding characteristic element \( \chi \) always satisfies the condition
\begin{equation} \label{Nec}
\sum_{l,I} (-1)^{|I|} \restr{ \widehat{D}_I \left( \frac{\partial F^l}{\partial u_I^k} \chi_l \right) }{\mathcal{E}} = 0
\end{equation}
for all \( k = 1, \dots, m \). This result reflects the cohomological origin of the characteristic element.

\begin{remark}
Let \( \mathfrak{g} = \mathbb{R} \), that is, the \( \mathfrak{g} \)-valued ZCR uniquely corresponds to a conserved current for \( \mathcal{E} \).  
In this case, the \textit{characteristic element} \( \chi \) coincides with the generating function \eqref{gen-funct} of the corresponding conservation law, and condition~\eqref{Nec} reduces to condition~\eqref{cosym}.
\end{remark}
For evolutionary equations, the characteristic element was discovered independently -- without any cohomological interpretation -- in \cite{Sa95}. \\

Let \( \mathcal{G} \subset GL(n) \) be a connected matrix Lie group with corresponding matrix Lie algebra \( \mathfrak{g} \).~A \textit{\( \mathcal{G} \)-valued function} on \( J^{\infty}(\pi) \) is a smooth function \( H: J^{\infty}(\pi) \to \mathcal{G} \) that depends on only finitely many variables \( u_I^k \).  
Such a function can be viewed as a matrix \( H \) whose entries are differential functions satisfying \( H(\theta) \in \mathcal{G} \) for all \( \theta \in J^{\infty}(\pi) \).  
The set of all \( \mathcal{G} \)-valued functions is denoted by \( \mathcal{F}^{\mathcal{G}}(\pi) \), and it naturally inherits the group structure under point-wise operations:
\begin{equation*}
H^{-1}(\theta):=\left(H(\theta)\right)^{-1},\ (H_1\cdot H_2)(\theta):=H_1(\theta)\cdot H_2(\theta)\mbox{ for all }\theta\in J^{\infty}(\pi).\end{equation*}
Obviously, we have the inclusion $\mathcal{F}^{\mathcal{G}}(\pi)\subset\mathcal{F}^{\mathfrak{gl}(n)}(\pi)$; hence the total derivative operators $D_x,D_y: \mathcal{F}^{\mathcal{G}}(\pi)\to\mathcal{F}^{\mathfrak{gl}(n)}(\pi)$ are well defined. Differentiating the identity $H\cdot H^{-1}=E$ yields
\begin{equation}\label{der-inv}D_x(H^{-1})=-H^{-1}\cdot D_x(H)\cdot H^{-1},\ \ D_y(H^{-1})=-H^{-1}\cdot D_y(H)\cdot H^{-1}.\end{equation}

Furthermore, a standard result from the theory of Lie groups and Lie algebras implies that the Lie algebra \(\mathcal{F}^\mathfrak{g}(\pi) \) is closed under conjugation by elements of the corresponding Lie group \( \mathcal{F}^\mathcal{G}(\pi) \). Similarly, it is well-known that for any $H\in\mathcal{F}^{\mathcal{G}}(\pi)$, the elements $(D_x H) H^{-1}$ and $(D_y H) H^{-1}$ belong to $\mathcal{F}^{\mathfrak{g}}(\pi)$, as they arise as coefficients of the right-translation of the horizontal differential of the map $H:J^{\infty}(\pi)\to\mathcal{G}$. For a more detailed discussion, see, for instance, \cite{Ma92}.

Hence
any \( \mathcal{G} \)-valued function \( H \in \mathcal{F}^{\mathcal{G}}(\pi) \) induces a \textit{gauge transformation} 
\begin{equation*}\mathcal{F}^{\mathfrak{g}}(\pi) \times \mathcal{F}^{\mathfrak{g}}(\pi)\to\mathcal{F}^{\mathfrak{g}}(\pi) \times \mathcal{F}^{\mathfrak{g}}(\pi),\ (A,B)\mapsto (A^H,B^H)\end{equation*}
given by
\begin{equation}\label{gauge-def}
A^H = D_x(H) H^{-1} + H A H^{-1}, \qquad
B^H = D_y(H) H^{-1} + H B H^{-1}.
\end{equation}

Passing to equivalence classes, this yields the transformation
\begin{equation*} t_H :\mathcal{F}^{\mathfrak{g}}(\mathcal{E}) \times \mathcal{F}^{\mathfrak{g}}(\mathcal{E})\to\mathcal{F}^{\mathfrak{g}}(\mathcal{E}) \times \mathcal{F}^{\mathfrak{g}}(\mathcal{E}),\  
 \left( \restr{A}{\mathcal{E}}, \restr{B}{\mathcal{E}} \right) \mapsto \left( \restr{A^H}{\mathcal{E}}, \restr{B^H}{\mathcal{E}} \right),
\end{equation*}
as witnessed by the following lemma.

\begin{lemma} For every $H \in \mathcal{F}^{\mathcal{G}}(\pi)$ the map
\begin{equation*} t_H :\mathcal{F}^{\mathfrak{g}}(\mathcal{E}) \times \mathcal{F}^{\mathfrak{g}}(\mathcal{E})\to\mathcal{F}^{\mathfrak{g}}(\mathcal{E}) \times \mathcal{F}^{\mathfrak{g}}(\mathcal{E}),\  
 \left( \restr{A}{\mathcal{E}}, \restr{B}{\mathcal{E}} \right) \mapsto \left( \restr{A^H}{\mathcal{E}}, \restr{B^H}{\mathcal{E}} \right)\end{equation*} 
 defined by \eqref{gauge-def} is well-defined; that is, it depends only on the equivalence classes \(\restr{A}{\mathcal{E}}\) and \(\restr{B}{\mathcal{E}}\), not on the particular representatives \(A,B\in\mathcal{F}^{\mathfrak{g}}(\pi)\).
\end{lemma}
\begin{proof}
 Suppose that \( \restr{A}{\mathcal{E}} = \restr{\tilde{A}}{\mathcal{E}} \) and \( \restr{B}{\mathcal{E}} = \restr{\tilde{B}}{\mathcal{E}} \), that is, \( A = \tilde{A} + X \) and \( B = \tilde{B} + Y \) for some \( X, Y \in \mathcal{I}^{\mathfrak{g}}(\mathcal{E}) \).  
Then we have
\begin{equation*}
A^H - \tilde{A}^H = (D_xH)H^{-1}+H(\widetilde A+X)H^{-1}-(D_xH)H^{-1}-H\widetilde A H^{-1} =H X H^{-1}  \in \mathcal{I}^{\mathfrak{g}}(\mathcal{E})
\end{equation*} and
\begin{equation*}B^H - \tilde{B}^H = (D_yH)H^{-1}+H(\widetilde B+Y)H^{-1}-(D_yH)H^{-1}-H\widetilde B H^{-1} =H Y H^{-1}  \in \mathcal{I}^{\mathfrak{g}}(\mathcal{E})
\end{equation*}
That is, \( \restr{A^H}{\mathcal{E}} = \restr{\tilde{A}^H}{\mathcal{E}} \) and \( \restr{B^H}{\mathcal{E}} = \restr{\tilde{B}^H}{\mathcal{E}} \).
\end{proof}
The statement of the next lemma is well known; we include it, along with a proof, purely for the sake of completeness.
\begin{lemma}
Let $ \restr{(A,B)}{\mathcal{E}} $ be a $ \mathfrak{g} $-valued zero-curvature representation for $ \mathcal{E} $, and let $ H \in \mathcal{F}^{\mathcal{G}}(\pi) $.  
Then the pair  
\begin{equation*}
\left( \restr{A^{H}}{\mathcal{E}}, \; \restr{B^{H}}{\mathcal{E}} \right),
\end{equation*}
with
\begin{equation*}
A^{H} = D_{x}(H) H^{-1} + H A H^{-1}, \qquad
B^{H} = D_{y}(H) H^{-1} + H B H^{-1},
\end{equation*}
is again a $ \mathfrak{g} $-valued zero-curvature representation for $ \mathcal{E} $.
\end{lemma}
\begin{proof}
By straightforward computations performed in $\mathcal{F}^{\mathfrak{gl}(n)}(\pi)$ (we omit the dots that indicate matrix multiplication) and by the use of the identities \eqref{der-inv} we have
\begin{align*}D_y(A^H)&=(D_yD_xH)H^{-1}+(D_x H)D_y(H^{-1})+(D_yH)AH^{-1}+H(D_yA)H^{-1}+HAD_y(H^{-1})\\
D_x(B^H)&=(D_xD_yH)H^{-1}+(D_y H)D_x(H^{-1})+(D_xH)BH^{-1}+H(D_xB)H^{-1}+HBD_x(H^{-1})\\
[A^H,B^H]&=(D_xH)H^{-1}(D_yH)H^{-1}+(D_xH)BH^{-1}+HAH^{-1}(D_yH)H^{-1}+HABH^{-1}\\
&-(D_yH)H^{-1}(D_xH)H^{-1}-(D_yH)AH^{-1}-HBH^{-1}(D_xH)H^{-1}-HBAH^{-1}\\
&=-(D_xH)(D_y(H^{-1}))+(D_xH)BH^{-1}-HA(D_y(H^{-1}))+HABH^{-1}\\
&+(D_yH)(D_x(H^{-1}))-(D_yH)AH^{-1}+HB(D_x(H^{-1}))-HBAH^{-1}\end{align*}
We thus have
\begin{equation*}D_yA^H-D_xB^H+[A^H,B^H]=H(D_yA-D_xB+[A,B])H^{-1}.\end{equation*}
Since, by assumption, $D_yA-D_xB+[A,B]$ belongs to $\mathcal{I}^{\mathfrak{g}}(\mathcal{E})$, and since $\mathcal{F}^{\mathfrak{g}}(\pi)$ is closed under the conjugation by elements from $\mathcal{F}^{\mathcal{G}}(\pi)$, it follows that the 
$\mathfrak{gl}(n)$-valued function $H(D_yA-D_xB+[A,B])H^{-1}$ belongs to $\mathcal{I}^{\mathfrak{g}}(\mathcal{E})$.\end{proof}

Each gauge transformation  thus preserves  the set of all \( \mathfrak{g} \)-valued ZCRs of \( \mathcal{E} \) invariant, and therefore, the formulas~\eqref{gauge-def} define an \textit{action }of the group $\mathcal{F}^{\mathcal{G}}(\pi)$ on the set of all \( \mathfrak{g} \)-valued ZCRs for \( \mathcal{E} \).

A \( \mathfrak{g} \)-valued ZCR \( \restr{(A, B)}{\mathcal{E}} \) for \( \mathcal{E} \) is said to be \textit{gauge equivalent} to the ZCR \( \restr{(\bar{A}, \bar{B})}{\mathcal{E}} \) if there exists a \( \mathcal{G} \)-valued function \( H \in \mathcal{F}^{\mathcal{G}}(\pi) \) such that
\begin{equation*}
\restr{\bar{A}}{\mathcal{E}} = \restr{A^H}{\mathcal{E}} \quad \text{and} \quad \restr{\bar{B}}{\mathcal{E}} = \restr{B^H}{\mathcal{E}},
\end{equation*}
where \( A^H \) and \( B^H \) are given by the formulas~\eqref{gauge-def}.
A ZCR that is gauge-equivalent to the zero ZCR \( (0, 0) \) is referred to as \textit{trivial}.  
That is, trivial \( \mathfrak{g} \)-valued ZCRs are precisely those of the form
\begin{equation*}
\left( \restr{D_x(H)H^{-1}}{\mathcal{E}},\ \restr{D_y(H)H^{-1}}{\mathcal{E}} \right),
\end{equation*}
where \( H \in \mathcal{F}^{\mathcal{G}}(\pi) \).

\begin{remark}\label{gaugeCL1}
In the case where \( \mathfrak{g} = \mathbb{R} \) and \( \mathcal{G} = \mathbb{R}^+ \), given \( H \in \mathcal{F}^{\mathbb{R}^+}(\pi) \), the corresponding gauge transformation reads as follows:
\begin{equation}\label{gaugeCL}
(\restr{A^H}{\mathcal{E}}, \restr{B^H}{\mathcal{E}}) = \left( \restr{D_x(H) H^{-1} + A}{\mathcal{E}},\ \restr{D_y(H) H^{-1} + B}{\mathcal{E}} \right)
= \left( \restr{D_x(\ln H) + A}{\mathcal{E}},\ \restr{D_y(\ln H) + B}{\mathcal{E}} \right),
\end{equation}
where the first equality follows from the commutativity of multiplication in \( \mathbb{R} \), and the second one from the fact that $H$ take values in $\mathbb{R}^+$.

On the other hand, if \( (\restr{A}{\mathcal{E}}, \restr{B}{\mathcal{E}}) \) is an \( \mathbb{R} \)-valued ZCR, then the pair \( (\restr{-B}{\mathcal{E}}, \restr{A}{\mathcal{E}}) \) corresponds to a conserved current. Let us consider the natural action of the additive group \( \mathcal{F}(\pi) \) on the set of conserved currents. Given \( R \in \mathcal{F}(\pi) \), the transformed conserved current is represented by the pair
\begin{equation*}
((-B)^R, A^R) = (-B - D_y R, A + D_x R).
\end{equation*}
Consequently, the associated transformed ZCR is given by the pair \( (A + D_x R, B + D_y R) \). Thus, the map
\begin{equation}\label{actZCR}
t_R \colon (\restr{A}{\mathcal{E}}, \restr{B}{\mathcal{E}}) \mapsto (A + D_x R, B + D_y R)
\end{equation}
defines an action of the additive group \( \mathcal{F}(\pi) \) on the set of all \( \mathbb{R} \)-valued ZCRs.

Note that this  group action is equivalent to the action of the multiplicative group $\mathcal{F}^{\mathbb{R}^+}(\pi)$    of positive-valued functions via gauge transformations \eqref{gaugeCL}. The two actions are related by the isomorphism
\begin{equation*}
\mathcal{F}^{\mathbb{R}^+}(\pi) \to \mathcal{F}^{\mathbb{R}}(\pi), \quad H \mapsto \ln H,
\end{equation*}
and hence yield the same orbits. In this sense, gauge-equivalent \( \mathbb{R} \)-valued ZCRs correspond to equivalent conservation laws.

\end{remark}

Finally, we recall the transformation behaviour of the characteristic element under gauge transformations of a ZCR. As shown in \cite{Ma92}, the characteristic element $\chi$ transforms by conjugation with the element $H$ when the ZCR is transformed via the gauge transformation $t_H$. More precisely, let $ \chi = \left( \restr{K_1}{\mathcal{E}}, \dots, \restr{K_N}{\mathcal{E}} \right) $ be the characteristic element of a $ \mathfrak{g} $-valued ZCR $ \restr{(A,B)}{\mathcal{E}} $ and let $ H \in \mathcal{F}^{\mathcal{G}}(\pi) $. Then the characteristic element 
\begin{equation*}
\chi^H = \left( \restr{K_1^H}{\mathcal{E}}, \dots, \restr{K_N^H}{\mathcal{E}} \right) \in \left( \mathcal{F}^{\mathfrak{g}}(\mathcal{E}) \right)^N
\end{equation*}
associated with the gauge-transformed ZCR $ \restr{(A^H, B^H)}{\mathcal{E}} $ satisfies
\begin{equation*}
\restr{K_j^H}{\mathcal{E}} = \restr{H K_j H^{-1}}{\mathcal{E}} \quad \text{for all } j = 1, \dots, N.
\end{equation*}

In contrast to the generating functions of conservation laws -- where each generating function uniquely determines the corresponding equivalence class of conserved currents (as defined in our setting) -- the situation is different for characteristic elements of zero-curvature representations. Two ZCRs may share the same characteristic element while not being gauge-equivalent (cf.~\cite{Ma92,Ma97}).
\begin{remark}
Since zero-curvature representations are typically studied up to gauge equivalence, the transformation property of the characteristic element described above is frequently exploited when searching for ZCRs of a given equation. The standard approach proceeds by first identifying all possible normal forms of the characteristic element under conjugation. One then seeks zero-curvature representations corresponding to these normal forms. All other ZCRs are gauge equivalent to one of these representatives. All such computations are carried out in internal coordinates on the given equation (see e.g. \cite{Ma92,Ma97, Ma02, Ma02-2,Se05,Se08}).
\end{remark}
\begin{remark}
In the general theory, a \( \mathfrak{g} \)-valued zero-curvature representation is defined as a horizontal differential \( 1 \)-form \( \alpha \) with coefficients in \( \mathfrak{g} \) that satisfies the Maurer--Cartan condition:
\begin{equation*}
\bar{d}\alpha = \frac{1}{2}[\alpha,\alpha],
\end{equation*}
where \( \bar{d} \) denotes the horizontal differential in the graded algebra $\mathfrak{g}\otimes\bar\Lambda\mathcal{E}$ of horizontal forms on $\mathcal{E}$ with coefficients in $\mathfrak{g}$. In coordinates, the form \( \alpha \) takes the shape \( \alpha = A\, \mathrm{d}x + B\, \mathrm{d}y \), where \( A \) and \( B \) are \( \mathfrak{g} \)-valued functions on \( \mathcal{E} \), and the Maurer--Cartan condition is then just \eqref{ZCR}. 

To each fixed $\mathfrak{g}$-valued ZCR $\alpha$, one can associate the so-called linear gauge complexes. The ZCR itself is then an element of the $0$th linear complex. Its image under the vertical differential is an element of the first linear complex and  uniquely determines a certain cohomology class. A particular isomorphic image of this class is the characteristic element, see \cite{Ma92} for more details.

In the case when $\alpha$ is an $\mathbb{R}$-valued ZCR -- that is, a conserved current -- the above-mentioned gauge complexes coincide with the columns of the zeroth `sheet' of the $\mathcal{C}$-spectral sequence for $\mathcal{E}$. This underlines the analogy with conservation laws.
\end{remark}

\section{The characteristic form of the ZCR representation}\label{sec:char-form}
In the preceding sections, we have developed the theory of conservation laws and the theory of ZCRs for equations in two independent variables in a manner that emphasizes the well-known fact that there is a one-to-one correspondence between conserved currents (resp. their equivalence classes)  and $\mathbb{R}$-valued ZCRs (resp. their equivalence classes under the gauge transformation), with the characteristic element of an $\mathbb{R}$-valued ZCR coinciding with the generating function of the corresponding conservation law.

The aim of this section is to (at least partially) address the question of whether the classical results concerning conservation laws, formulated in Section~2, can be extended to $\mathfrak{g}$-valued ZCRs for an arbitrary matrix Lie algebra $\mathfrak{g} \subset \mathfrak{gl}(n)$. Since the statements of Classical Result~2 and Classical Result~3 are direct consequences of Classical Result~1, we expect that if they can be extended at all to $\mathfrak{g}$-valued ZCRs, such extensions will follow as consequences of an appropriate generalization of Classical Result~1.

The following proposition can thus be regarded as the main result of this paper, as it extends the fundamental statement of Classical Result~\ref{prop1} on the characteristic form of conservation laws to the setting of  $\mathfrak{g}$-valued zero-curvature representations. A result of a similar nature for evolutionary equations -- formulated in internal coordinates on the equation manifold -- was obtained in \cite{Sa95}. In that work, the characteristic element of a ZCR was introduced as the function that emerges when the Maurer-Cartan condition is cast into a particular canonical form in those coordinates. Our present result confirms that, even for a general (not necessarily evolutionary) equation, the characteristic element of a ZCR still carries essential jet coordinate-level structural information about the ZCR.
\begin{proposition}\label{CharForm}
Let $\restr{(A,B)}{\mathcal{E}}$ be a $\mathfrak{g}$-valued ZCR. 
There exists its extension $(\tilde A,\tilde B)$ to  $J^{\infty}(\pi)$ and an $N$-tuple $(Q_1,\ \dots,Q_N)\in{(\mathcal{F}^{\mathfrak{g}}(\pi))}^{N}$ of $\mathfrak{g}$-valued functions on $J^{\infty}(\pi)$, 
such that the condition
\begin{equation}\label{charform}
D_y(\widetilde{A})-D_x(\widetilde{B})+[\widetilde A,\widetilde B]=\sum_{l=1}^N F^lQ_l 
\end{equation}
is satisfied on $J^{\infty}(\pi)$.\\
Moreover, the restriction $\restr{Q}{\mathcal{E}}=(\restr{Q_1}{\mathcal{E}},\dots,\restr{Q_N}{\mathcal{E}})$ is precisely the characteristic element of the $\mathfrak{g}$-valued ZCR $\restr{(A,B)}{\mathcal{E}}=\restr{(\widetilde A,\widetilde B)}{\mathcal{E}}$.

\end{proposition}

\begin{proof}
Firstly, note that the following identities are true on $J^{\infty}(\pi)$  for any $\mathfrak{g}$-valued functions $M,R,S\in\mathcal{F}^{\mathfrak{g}}(\pi)$ and any differential function $f\in\mathcal{F}(\pi)$ (see also \cite{Sa95}):
\begin{equation}\label{sa}D_x(f)\cdot M=\widehat D_x^R(f\cdot M)-f\cdot \widehat D_x^RM,\ \ \ D_y(f)\cdot M=\widehat D_y^S(f\cdot M)-f\cdot \widehat D_y^SM.\end{equation}
By induction we get the identities

\begin{equation*}
(D_x^af)\cdot M=\widehat D_x^R\left(\sum_{j=1}^{a}(-1)^{j-1}((D_x)^{a-j}f)\cdot((\widehat D_x^R)^{j-1} M)\right)+(-1)^af\cdot((\widehat D_x^R)^aM)
\end{equation*}
\begin{equation*}
(D_y^bf)\cdot M=\widehat D_y^S\left(\sum_{j=1}^{b}(-1)^{j-1}((D_y)^{b-j}f)\cdot((\widehat D_y^S)^{j-1} M)\right)+(-1)^bf\cdot((\widehat D_y^S)^bM)
\end{equation*}
that are true for all $a,b\in\mathbb{N}$, $a,b\geq 1$, hence,
for any such $a,b$ we have
\begin{align}
\lefteqn{(D_x^aD_y^bf)\cdot M}\nonumber\\
&=\left(D_x^a\left(D_y^bf\right)\right)\cdot M\nonumber\\
&=\widehat D_x^R\left(\sum_{j=1}^{a}(-1)^{j-1}((D_x)^{a-j}(D_y^bf))\cdot((\widehat D_x^R)^{j-1} M)\right)+(-1)^a(D_y^bf)\cdot((\widehat D_x^R)^aM)\nonumber\\
&=\widehat D_x^R\left(\sum_{j=1}^{a}(-1)^{j-1}((D_x)^{a-j}(D_y^bf))\cdot((\widehat D_x^R)^{j-1} M)\right)+\widehat D_y^S\left(\sum_{i=1}^{b}(-1)^{a+i-1}((D_y)^{b-i}f)\cdot((\widehat D_y^S)^{i-1} (\widehat D_x^R)^aM)\right)\nonumber\\
&+(-1)^{a+b}f\cdot((\widehat D_y^S)^b(\widehat D_x^R)^aM),\nonumber
\end{align}
where $R,S,M\in\mathcal{F}^{\mathfrak{g}}(\pi)$ are arbitrary $\mathfrak{g}$-valued functions on $J^{\infty}(\pi)$ and $f\in\mathcal{F}(\pi)$ is an arbitrary differential function on $J^{\infty}(\pi)$.

Now, let \( \restr{(A,B)}{\mathcal{E}} \) be a \( \mathfrak{g} \)-valued ZCR for \( \mathcal{E} \), and let \( (A,B) \) be an arbitrary representative of this equivalence class. Consider the extension of the Maurer--Cartan condition for \( (A,B) \) to \( J^{\infty}(\pi) \), as given by equality~\eqref{ZCR2}. Employing the identities derived above, we can now rewrite the right-hand side as follows:

\begin{align}
\sum_{I,l}D_I(F^l)\cdot C_l^I
&=\sum_{l=1}^N\left(F^l\cdot C_I^{(0,0)}+\sum_{a\geq 1} D_x^a(F^l)\cdot C_l^{(a,0)}+\sum_{b\geq 1} D_y^b(F^l)\cdot C_l^{(0,b)}+\sum_{a,b\geq 1}D_x^a(D_y^b(F^l))\cdot C_l^{(a,b)}\right)\nonumber\\
&=\sum_{l=1}^NF^l\cdot C_I^{(0,0)}\nonumber\\
&+\sum_{l=1}^N\sum_{a\geq 1}\left(
\widehat D_x^R\left(\sum_{j=1}^{a}(-1)^{j-1}((D_x)^{a-j}F^l)\cdot((\widehat D_x^R)^{j-1} C_l^{(a,0)})\right)+(-1)^aF^l\cdot((\widehat D_x^R)^aC_l^{(a,0)})
\right)\nonumber\\
&+\sum_{l=1}^N\sum_{b\geq 1}\left(
\widehat D_y^S\left(\sum_{j=1}^{b}(-1)^{j-1}((D_y)^{b-j}F^l)\cdot((\widehat D_y^S)^{j-1} C_l^{(0,b)})\right)+(-1)^bF^l\cdot((\widehat D_y^S)^bC_l^{(0,b)})
\right)\nonumber\\
&+\sum_{l=1}^N\sum_{a,b\geq 1}\left[\widehat D_x^R\left(\sum_{j=1}^{a}(-1)^{j-1}((D_x)^{a-j}(D_y^bF^l))\cdot((\widehat D_x^R)^{j-1} C_l^{(a,b)})\right)\right.\nonumber\\
&\left.+\widehat D_y^S\left(\sum_{i=1}^{b}(-1)^{a+i-1}((D_y)^{b-i}F^l)\cdot((\widehat D_y^S)^{i-1} (\widehat D_x^R)^aC_l^{(a,b)})\right)+(-1)^{a+b}F^l\cdot((\widehat D_y^S)^b(\widehat D_x^R)^aC_l^{(a,b)})\right]\nonumber\\
&=\widehat D_x^R\left(\sum_{l=1}^N\sum_{a\geq 1}\sum_{j=1}^{a}\sum_{b\geq 0}(-1)^{j-1}((D_x)^{a-j}(D_y^bF^l))\cdot((\widehat D_x^R)^{j-1} C_l^{(a,b)})\right)\label{B1}\\
&+\widehat D_y^S\left(\sum_{l=1}^N\sum_{b\geq 1}\sum_{j=1}^{b}\sum_{a\geq 0}(-1)^{a+j-1}((D_y)^{b-j}F^l)\cdot((\widehat D_y^{S})^{j-1} (\widehat D_x^R)^aC_l^{(a,b)})\right)\label{A1}\\
&+\sum_{l=1}^N F^l\cdot\left(\sum_{a,b\geq 0}(-1)^{a+b}((\widehat D_y^S)^b(\widehat D_x^R)^aC_l^{(a,b)})\right).\label{formulas}
\end{align}

where $R,S\in\mathcal{F}^{\mathfrak{g}}(\pi)$ are still arbitrary $\mathfrak{g}$-valued functions on $J^{\infty}(\pi)$
(all the sums are finite, thus we are allowed to rearange their terms arbitrarily).

 Since~\eqref{formulas} holds for all such \( R \) and \( S \), it in particular holds for the specific choice \( R = A \) and \( S = B + B_1 \), where

\begin{equation*}B_1:=\sum_{l=1}^N\sum_{a\geq 1}\sum_{j=1}^{a}\sum_{b\geq 0}(-1)^{j-1}((D_x)^{a-j}(D_y^bF^l))\cdot((\widehat D_x^A)^{j-1} C_l^{(a,b)})\end{equation*}
($B_1$  as well as $B+B_1$ are still $\mathfrak{g}$-valued functions on $J^{\infty}(\pi)$, as follows from the properties of $\mathcal{F}^{\mathfrak{g}}(\pi)$ stated in Section \ref{sec:1}). We thus have
\begin{align}
\sum_{I,l}D_I(F^l)\cdot C_l^I&= \widehat D_x^A\left(\sum_{l=1}^N\sum_{a\geq 1}\sum_{j=1}^{a}\sum_{b\geq 0}(-1)^{j-1}((D_x)^{a-j}(D_y^bF^l))\cdot((\widehat D_x^A)^{j-1} C_l^{(a,b)})\right)\nonumber\\
&+\widehat D_y^{B+B_1}\left(\sum_{l=1}^N\sum_{b\geq 1}\sum_{j=1}^{b}\sum_{a\geq 0}(-1)^{a+j-1}((D_y)^{b-j}F^l)\cdot((\widehat D_y^{B+B_1})^{j-1} (\widehat D_x^A)^aC_l^{(a,b)})\right)\nonumber\\
&+\sum_{l=1}^N F^l\cdot\left(\sum_{a,b\geq 0}(-1)^{a+b}((\widehat D_y^{B+B_1})^b(\widehat D_x^A)^aC_l^{(a,b)})\right).\nonumber
\end{align}
Let us put 
\begin{equation*}A_1:=\sum_{l=1}^N\sum_{b\geq 1}\sum_{j=1}^{b}\sum_{a\geq 0}(-1)^{a+j-1}((D_y)^{b-j}F^l)\cdot((\widehat D_y^{B+B_1})^{j-1} (\widehat D_x^A)^aC_l^{(a,b)})\end{equation*}
and
\begin{equation*}Q_l:=\sum_{a,b\geq 0}(-1)^{a+b}((\widehat D_y^{B+B1})^b(\widehat D_x^A)^aC_l^{(a,b)}).\end{equation*}
The right-hand side of \eqref{ZCR2} can now be rewritten in the form
\begin{align}
\sum_{I,l}D_I(F^l)\cdot C_l^I
&=\widehat D_x^{A}(B_1)+\widehat D_y^{B+B_1}(A_1)+\sum_{l=1}^N F^l\cdot Q_l\nonumber\\
&=D_x(B_1)-[A,B_1]+D_y(A_1)-[B+B_1,A_1]+\sum_{l=1}^N F^l\cdot Q_l\nonumber\\
&=D_x(B_1)+D_y(A_1)+[A,B]-[A-A_1,B+B_1]+\sum_{l=1}^N F^l\cdot Q_l\nonumber
\end{align}
The equality \eqref{ZCR2} thus yields
\begin{equation*}
D_yA-D_xB+[A,B]=D_x(B_1)+D_y(A_1)+[A,B]-[A-A_1,B+B_1]+\sum_{l=1}^N F^l\cdot Q_l\ \mbox{ identically on } J^{\infty}(\pi).
\end{equation*}
and it is equivalent to
\begin{equation}\label{eqZCR}
D_y(A-A_1)-D_x(B+B_1)+[A-A_1,B+B_1]=\sum_{l=1}^N F^l\cdot Q_l \ \mbox{ identically on } J^{\infty}(\pi).
\end{equation} 
The pair \( (\widetilde{A}, \widetilde{B}) \), where \( \widetilde{A} := A - A_1 \) and \( \widetilde{B} := B + B_1 \), can be immediately recognized as a representative of a \( \mathfrak{g} \)-valued ZCR for \( \mathcal{E} \), since it  satisfies 
\begin{equation*}D_y\widetilde{A}-D_x\widetilde{B}+[\widetilde A,\widetilde B]=\sum_{l=1}^N F^l\cdot Q_l\end{equation*} identically on $J^{\infty}(\pi)$.
Moreover, the pair $(\widetilde{A}, \widetilde{B}) \in \mathcal{F}^{\mathfrak{g}}(\pi) \times \mathcal{F}^{\mathfrak{g}}(\pi)$ is equivalent to the pair $(A, B) \in \mathcal{F}^{\mathfrak{g}}(\pi) \times \mathcal{F}^{\mathfrak{g}}(\pi)$, since
$-A_1$ and  $ B_1$ are,
according to Lemma \ref{decomp_lm2}, elements of the ideal $\mathcal{I}^{\mathfrak{g}}(\mathcal{E})$. Thus, we have $(\restr{A}{\mathcal{E}},\restr{B}{\mathcal{E}})=(\restr{\widetilde A}{\mathcal{E}},\restr{\widetilde B}{\mathcal{E}})$, which means that $(\widetilde A,\widetilde B)$ is the desired extension of $\restr{(A,B)}{\mathcal{E}}$.

Finally, $\restr{Q}{\mathcal{E}} = (\restr{Q_1}{\mathcal{E}}, \dots, \restr{Q_N}{\mathcal{E}})$ is the characteristic element of $(\restr{\widetilde A}{\mathcal{E}},\restr{\widetilde B}{\mathcal{E}})$ (and hence of $(\restr{A}{\mathcal{E}}, \restr{B}{\mathcal{E}})$) by the defining condition~\eqref{char-el}. This completes the proof.

\end{proof}

\begin{remark}
Let us verify that $\restr{Q}{\mathcal{E}} = (\restr{Q_1}{\mathcal{E}}, \dots, \restr{Q_N}{\mathcal{E}})$ coincides with the characteristic element of the  $\mathfrak{g}$-valued ZCR $\restr{(A,B)}{\mathcal{E}}$,  when computed directly from its decomposition~\eqref{ZCR2}. According to formula~\eqref{char-el}, the characteristic element of $\restr{(A,B)}{\mathcal{E}}$ is given by $\restr{K}{\mathcal{E}} = (\restr{K_1}{\mathcal{E}}, \dots, \restr{K_N}{\mathcal{E}})$, where
$
K_l = \sum_I (-1)^{|I|} \widehat{D}_I C_l^I.
$
But, we have
\begin{align*}\restr{Q_l}{\mathcal{E}}&=\sum_{a,b\geq 0}(-1)^{a+b}\restr{((\widehat D_y^{B+B1})^b(\widehat D_x^A)^aC_l^{(a,b)})}{\mathcal{E}}=\sum_{a,b\geq 0}(-1)^{a+b}(\widehat D_y^{B+B1})^b(\widehat D_x^A)^a(\restr{C_l^{(a,b)}}{\mathcal{E}})\\
&=\sum_{a,b\geq 0}(-1)^{a+b}(\widehat D_y^{B})^b(\widehat D_x^A)^a(\restr{C_l^{(a,b)}}{\mathcal{E}})=\sum_I(-1)^{|I|}\restr{\widehat{D}_I C_l^l}{\mathcal{E}},\\
\end{align*}
where the third equality follows from the fact that $\restr{B}{\mathcal{E}}=\restr{(B+B_1)}{\mathcal{E}}$ and from the properties of the operators $\widehat D_x^R$ and $\widehat D_y^S$ on $\mathcal{F}^{\mathfrak{g}}(\mathcal{E})$ discussed in Section \ref{sec:1}.\\ 
This observation exactly corresponds to the fact that the characteristic element is an object associated with the entire equivalence class $\restr{(A,B)}{\mathcal{E}}$. 
\end{remark}

Following  terminology for conservation laws we refer to the Maurer--Cartan condition on $\restr{(A,B)}{\mathcal{E}}$ written in the form
\begin{equation}\label{charZCR}D_yA-D_xB+[A,B]=\sum_{l=1}^NF^lQ_l \mbox{ on } J^{\infty}(\pi)\end{equation} as to the \textit{characteristic form of the Maurer--Cartan condition}, the two-tuple $(A,B)\in\mathcal{F}^{\mathfrak{g}}(\pi)\times\mathcal{F}^{\mathfrak{g}}(\pi)$ realizing the characteristic form \eqref{charZCR} will be called the \textit{characteristic representative of the ZCR}. The corresponding $N$-tuple $Q=(Q_1,\dots,Q_N)\in(\mathcal{F}^{\mathfrak{g}}(\pi))^{N}$ of $\mathfrak{g}$-valued functions is called \textit{the characteristic} of the $\mathfrak{g}$-valued ZCR $\restr{(A,B)}{\mathcal{E}}$ corresponding to the characteristic representative $(A,B)$.

\begin{remark}
Let us note that the procedure for finding a characteristic representative presented in the proof of Proposition~\ref{CharForm} is not the only possible one. Observe that the resulting form depends on the order in which we eliminate the mixed derivatives \( (D_x)^a (D_y)^b(F) \) -- this order determines the order in which the remaining terms are acted upon by the operators \( \widehat{D}_x^R \) and \( \widehat{D}_y^S \), which, unlike \( D_x \) and \( D_y \), do not commute on \( J^{\infty}(\pi) \).

Moreover, the characteristic form also depends on the initial decomposition chosen for the ZCR condition. Different characteristic representatives will, of course, yield different characteristics. However, all of them are equivalent, meaning they represent the same characteristic element.
 \end{remark}

In the following example, we apply the procedure outlined in the proof of Proposition~\ref{CharForm} to two distinct representatives of a single \( \mathfrak{sl}(2) \)-valued ZCR for the KdV equation. This yields two distinct characteristic representatives with two distinct characteristics. However, we will see that both of these characteristics are extensions of the same $\mathfrak{g}$-valued function on $\mathcal{E}$ - the characteristic element of the ZCR.  This example thus illustrates that characteristics associated with two distinct characteristic representatives of a given ZCR are, in fact, equivalent.

\begin{example}\label{ex1}
Let $\mathcal{E}$ be given by the KdV equation and all its differential consequences
$$\mathcal{E}:\ F=u_y-u_{xxx}+6uu_x=0$$
(to remain consistent with the preceding notation, we write the independent variable usually denoted by \( t \) as \( y \)
). \\
Consider the matrices
\begin{equation*}A=\begin{pmatrix}
6uu_{xx}+6u_x^2-u_{4x}+u_{xy}&1\\u&-6uu_{xx}-6u_x^2+u_{4x}-u_{xy}
\end{pmatrix}\end{equation*}
 \mbox{ and } 
\begin{equation*}B=\begin{pmatrix}
6uu_x+u_x-u_{xxx}+u_y&-2u\\-2u^2+u_{xx}&-6uu_x-u_x+u_{xxx}-u_y
\end{pmatrix}
\end{equation*}
It can be verified by direct computation that the pair of matrices \( (A, B) \)  satisfies
\begin{equation*}D_yA-D_xB+[A,B]=\begin{pmatrix*}-D_xF+D_x(D_yF)&-2F-4uD_xF\\(2u+1)F+(4u^2-2u_{xx})D_xF&D_xF-D_x(D_yF)\end{pmatrix*}\end{equation*}
identically on $J^{\infty}(\pi)$.
Since each entry of the matrix on the right-hand side lies in $\mathcal{I}(\mathcal{E})$, we conclude that the entire matrix belongs to $\mathcal{I}^{\mathfrak{sl}(2)}(\mathcal{E})$, and we have $\restr{(D_y A - D_x B + [A, B])}{\mathcal{E}} = 0$. Therefore, the pair $(A,B)$ defines an $\mathfrak{sl}(2)$-valued ZCR for $\mathcal{E}$.

Let us consider the decomposition 
\begin{equation*}D_yA-D_xB+[A,B]=F\begin{pmatrix}
0&-2\\2u+1&0
\end{pmatrix}
+D_xF\begin{pmatrix}
-1&-4u\\4u^2-2u_{xx}&1
\end{pmatrix}
+
D_x(D_yF)\begin{pmatrix}
1&0\\0&-1
\end{pmatrix}.\end{equation*}
For clarity, we adopt the same notation as used in the proof. We set
\begin{equation*}C^{(0,0)}:=\begin{pmatrix}
0&-2\\2u+1&0
\end{pmatrix}\ \ C^{(1,0)}:=\begin{pmatrix}
-1&-4u\\4u^2-2u_{xx}&1
\end{pmatrix}, \mbox{ and } C^{(1,1)}:=\begin{pmatrix}
1&0\\0&-1
\end{pmatrix}.\end{equation*}
Our aim is to employ identities \eqref{sa} involving arbitrary (but fixed) \( \mathfrak{g} \)-valued functions \( R \) and \( S \), so that the resulting expression on the right-hand side matches the desired structure:
\begin{equation*}
F \cdot (\text{something}) + \widehat{D}_x^R(\text{something}) + \widehat{D}_y^S(\text{something}).
\end{equation*}
Obviously, we have
\begin{equation*}D_xF\cdot C^{(1,0)}=\widehat D_x^R(F\cdot C^{(1,0)})-F\cdot\widehat D_x^R(C^{(1,0)}),\end{equation*} and, upon repeated application of the identities, we arrive at
\begin{equation*}D_x(D_y(F))\cdot C^{(1,1)}=\widehat D_x^{R}(D_yF\cdot C^{(1,1)})-\widehat D_y^S(F\cdot \widehat D_x^R(C^{(1,1)}))+F\cdot \widehat D_y^S(\widehat D_x^{R}(C^{(1,1)})).\end{equation*}
Thus, we have
\begin{align*}
\lefteqn{D_yA-D_xB+[A,B]}\\&=&F\cdot\left(C^{(0,0)}-\widehat D_x^R(C^{(1,0)})+\widehat D_y^S(\widehat D_x^{R}(C^{(1,1)}))\right)+\widehat D_x^R\left(F\cdot C^{(1,0)}+D_yF\cdot C^{(1,1)}\right)+\widehat D_y^S(-F\cdot\widehat D_x^{R}(C^{(1,1)})),
\end{align*}
where $R,S$ are still fixed arbitrary functions.
At this point, we apply the trick used in the proof of Proposition \eqref{CharForm} and we set 
\begin{equation*}R:=A \mbox{ and } S:=B+B_1,\ \mbox{ where }B_1:=F\cdot C^{(1,0)}+D_yF\cdot C^{(1,1)}.\end{equation*}
We thus have
\begin{align*}
\lefteqn{D_yA-D_xB+[A,B]}\\&=F\cdot\left(C^{(0,0)}-\widehat D_x^A(C^{(1,0)})+\widehat D_y^{B+B_1}(\widehat D_x^{A}(C^{(1,1)}))\right)+\widehat D_x^A\left(B_1\right)+\widehat D_y^{B+B1}(-F\cdot\widehat D_x^{A}(C^{(1,1)})).\end{align*}
Further, we put $A_1:=-F\cdot\widehat D_x^{A}(C^{(1,1)})$, $\widetilde Q:=C^{(0,0)}-\widehat D_x^A(C^{(1,0)})+\widehat D_y^{B+B_1}(\widehat D_x^{A}(C^{(1,1)}))$, which yields
\begin{align*}
D_yA-D_xB+[A,B]&=F\cdot \widetilde Q+\widehat D_x^A\left(B_1\right)+\widehat D_y^{B+B1}(A_1).\end{align*}
Finally, we apply the definition of the operators  $\widehat D_x^A$ and $\widehat D_y^{B+B1}$ and we obtain
\begin{align*}
D_yA-D_xB+[A,B]=F\widetilde Q+D_x B_1-[A,B_1]+D_yA_1-[B+B_1,A_1].\end{align*}
However, the equality last obtained is equivalent to
\begin{equation*}D_y(A-A_1)-D_x(B+B_1)+[A-A_1,B+B_1]=F\cdot\widetilde Q.\end{equation*}
The functions \( A_1 \) and \( B_1 \) can be explicitly computed, as they are defined in terms of known quantities; however, we omit their explicit forms here.  Instead, we present only the resulting formulas for the desired characteristic representative $(\widetilde A,\widetilde B)$ of the ZCR $(\restr{A}{\mathcal{E}},\restr{B}{\mathcal{E}})$:
\begin{equation*}\widetilde A=
\begin{pmatrix}
6uu_{xx}+6u_x^2-u_{4x}+u_{xy}&12uu_x-2u_{xxx}+2u_y+1\\
u-12u^2u_x+2uu_{xxx}-2uu_y&-6uu_{xx}-6u_x^2+u_{4x}-u_{xy}
\end{pmatrix}\end{equation*}
and \begin{equation*}
\widetilde B=\begin{pmatrix}
6uu_{xy} + 6u_xu_y + u_x - u_{xxxy} + u_{yy}&-2u - 4(6uu_x - u_{xxx} + u_y)u\\
-2u^2 + u_{xx} + (6uu_x - u_{xxx} + u_y)(4u^2 - 2u_{xx})&-6uu_{xy} - 6u_xu_y - u_x + u_{xxxy} - u_{yy}
\end{pmatrix}.
\end{equation*}
The corresponding characteristic is \begin{equation*}\widetilde Q=
\begin{pmatrix}

-4u_{xx}F&-8uD_xF-4D_yF\\
1-2F+(4u_xx-8u^2)D_xF-4uD_yF&4u_{xx}F
\end{pmatrix}.\end{equation*}
A direct computation confirms that \( (\widetilde A, \widetilde B) \) is indeed a characteristic representative, with the corresponding characteristic coinciding with \( \widetilde{Q} \).

Since the ZCR is defined as the set of all equivalent pairs of functions, the procedure from the proof of Proposition \ref{CharForm} could be applied to any other representative of the same ZCR -- namely, to any pair \( (\bar{A}, \bar{B}) \) obtained from \( (A, B) \) by adding an element of \( \mathcal{I}^{\mathfrak{g}}(\mathcal{E}) \).

Let us try to carry this out. Consider a pair of \( \mathfrak{g} \)-valued functions

\begin{equation*}\bar{A}=A+D_xF\cdot\begin{pmatrix*}[r]-1&0\\0&1\end{pmatrix*}=\begin{pmatrix}
0&1\\
u&0
\end{pmatrix}, \mbox{ and }
\bar B= B+F\cdot\begin{pmatrix*}[r]-1&0\\0&1\end{pmatrix*}=\begin{pmatrix}
 u_x&-2u\\
 -2u^2 + u_{xx}&-u_x
 \end{pmatrix}.\end{equation*}
 This pair is well-known to define an $\mathfrak{sl}(2)$-valued ZCR. A direct computation shows that 
\begin{equation*}D_y\bar A-D_x\bar B+[\bar A,\bar B]=F\cdot \begin{pmatrix}
 0&0\\
 1&0
 \end{pmatrix},\end{equation*}
 thus, 
 \( (\bar A, \bar B) \) itself is a characteristic representative of our ZCR, with the corresponding characteristic being the function
\begin{equation*}
 \bar Q=
 \begin{pmatrix}
 0&0\\
 1&0
 \end{pmatrix}.
\end{equation*}
Since 
\begin{equation*}\widetilde Q-\bar Q=\begin{pmatrix}

-4u_{xx}F&-8uD_xF-4D_yF\\
-2F+(4u_xx-8u^2)D_xF-4uD_yF&4u_{xx}F
\end{pmatrix}\in\mathcal{I}^{\mathfrak{sl}(2)}(\mathcal{E}),\end{equation*}
the two characteristics \( \widetilde{Q} \) and \( \bar{Q} \) are equivalent, i.e., $ \restr{\widetilde{Q}}{\mathcal{E}} = \restr{\bar{Q}}{\mathcal{E}}$ , which is exactly as expected.
\end{example}

Let us note that the result of the above transformation, which replaced the matrices \( A \) and \( B \) with the equivalent pair \( \bar{A} \) and \( \bar{B} \), coincides with the result that would be obtained by solving the equation \( F = 0 \) for the variable \( u_y \), and then substituting for \( u_y \) and its derivatives into \( A \) and \( B \). This procedure corresponds to the standard introduction of so-called internal coordinates on evolution equations. It may thus appear to offer a significantly faster method for constructing a characteristic representative of a given ZCR than the one described in the proof of Proposition~\ref{CharForm}.
However, internal coordinates for the KdV equation can be introduced in a different way: one may solve the equation for \( u_{xxx} \) and subsequentally substitute for $u_{xxx}$ and its derivatives into \( (A, B) \). Yet, the resulting pair of matrices in this case will no longer be a characteristic representative.

The reason lies in the fact that each such substitution, arising from a choice of internal coordinates, yields a specific representative of the ZCR \( (\restr{A}{\mathcal{E}}, \restr{B}{\mathcal{E}}) \). However, not every representative is a characteristic one: in the previous example, we presented three representatives of the same ZCR,  \( (\widetilde{A}, \widetilde{B}) \) and \( (\bar{A}, \bar{B}) \) were characteristic representatives, whereas \( (A, B) \) was not.  Since the space of representatives for a given ZCR is infinite, it is unlikely that a randomly chosen representative, that arises as the result of the choice of internal coordinates will be a characteristic representative.\\

The following corollary is an immediate consequence of the preceding proposition and extends the statement of Classical Result~\ref{prop3}, originally formulated in the context of conservation laws, to the broader setting of ZCRs.

\begin{corollary}\label{cor}
The $N$-tuple $\chi=(\chi_1,\dots,\chi_N)\in(\mathcal{F}^{\mathfrak{g}}(\mathcal{E}))^N$ is the characteristic element of the $\mathfrak{g}$-valued ZCR $(\restr{A}{\mathcal{E}},\restr{B}{\mathcal{E}})$ if and only if there exists an extension $(\widetilde A,\widetilde B)\in\mathcal{F}^{\mathfrak{g}}(\pi)$  of $(\restr{A}{\mathcal{E}},\restr{B}{\mathcal{E}})$ and an extension $Q=(Q_1,\dots,Q_N)\in\mathcal{F}^{\mathfrak{g}}(\pi)$ of $\chi$ to $J^{\infty}(\pi)$ such that
\begin{equation*}D_y(\widetilde{A})-D_x(\widetilde{B})+[\widetilde A,\widetilde B]=\sum_{l=1}^N F^lQ_l \end{equation*} identically on $J^{\infty}(\pi)$.
\end{corollary}

The next proposition presents a natural generalization of Classical Result~\ref{prop4}, transferring its content from the framework of conservation laws to that of $\mathfrak{g}$-valued ZCRs, assuming $\mathfrak{g}$ is abelian.

\begin{proposition}\label{ZCRprop4}
Let $\mathfrak{g}\subset\mathfrak{gl}(n)$ be an abelian matrix Lie algebra. The $N$-tuple of functions \( \psi = (\psi_1, \dots, \psi_N) \in(\mathcal{F}^{\mathfrak{g}}(\mathcal{E}) )^N\) is the characteristic element of a  $\mathfrak{g}$-valued ZCR for $\mathcal{E}$ if and only if  there exists its extension   
 \( (Q_1, \dots, Q_N) \in (\mathcal{F}^{\mathfrak{g}}(\pi))^N \) that satisfies
 \begin{equation}\label{ZCRprop4char_suf}
 0=\sum_{l=1}^N\sum_I(-1)^{|I|}D_I\left(\frac{\partial F^l}{\partial u_I^k}Q_l\right)+\sum_{l=1}^N\sum_I(-1)^{|I|}D_I\left(\frac{\partial Q^l}{\partial u_I^k}F_l\right)
 \end{equation}
  identically on $J^{\infty}(\pi)$ for all $k=1,\dots,m$.
\end{proposition}

\begin{proof}
Let $\psi = (\psi_1, \dots, \psi_N) \in (\mathcal{F}^{\mathfrak{g}}(\mathcal{E}))^N$ be the characteristic element of a $\mathfrak{g}$-valued ZCR $(\restr{A}{\mathcal{E}}, \restr{B}{\mathcal{E}})$, and let $(\widetilde{A}, \widetilde{B})$ be its characteristic representative. According to Corollary~\ref{cor}, we know that $\psi$ is the restriction of the characteristic 
$Q = (Q_1, \dots, Q_N) \in (\mathcal{F}^{\mathfrak{g}})^N$ corresponding to $(\widetilde{A}, \widetilde{B})$, and these functions satisfy condition
\begin{equation*}
D_y \widetilde{A} - D_x \widetilde{B} + [\widetilde{A}, \widetilde{B}] = \sum_{l} F^l Q_l.
\end{equation*}
However, since $\mathfrak{g}$ is abelian, we have $[\widetilde{A}, \widetilde{B}] = 0$, so the identity simplifies to
\begin{equation*}
\sum_{l} F^l Q_l = D_y \widetilde{A} - D_x \widetilde{B} = \mathrm{Div}^{\mathfrak{g}}(-\widetilde{B}, \widetilde{A}).
\end{equation*}
Applying the $k$-th component $\mathrm{\mathbf{E}}_k^{\mathfrak{g}}$ of the $\mathfrak{g}$-analogue of the Euler operator to the above identity and using the result of Lemma~\ref{eulerdivLie}, we obtain the desired equality.

On the other hand, let \( (Q_1, \dots, Q_N) \in (\mathcal{F}^{\mathfrak{g}}(\pi))^N \) satisfy~\eqref{ZCRprop4char_suf}. Since this condition can be rewritten as
\begin{equation*}
\mathbf{E}_k^{\mathfrak{g}}\left(\sum_{l=1}^N F^l Q_l\right) = 0,
\end{equation*} for all $k=1,\dots,m$,
Lemma~\ref{eulerdivLie} implies that
\begin{equation*}
\sum_{l=1}^N F^l Q_l = \mathrm{Div}^{\mathfrak{g}}(A, B) = D_x A + D_y B = D_x A + D_y B - [A, B],
\end{equation*}
for some $A,B \in \mathcal{F}^{\mathfrak{g}}(\pi)$, where the last equality holds since $\mathfrak{g}$ is abelian. Therefore, according to the definition, the restriction $\restr{(Q_1, \dots, Q_N)}{\mathcal{E}}$ is the characteristic element of the $\mathfrak{g}$-valued ZCR $\restr{(B,-A)}{\mathcal{E}}$, which completes the proof.

\end{proof}
The situation turns out to be entirely different when $\mathfrak{g}$ is nonabelian.
The following proposition provides a necessary condition that must be satisfied by every characteristic representative of a zero-curvature representation. This condition is genuinely new -- it does not arise as an extension of any known nontrivial result concerning conserved currents.

\begin{proposition}\label{newNec}
Let $(\tilde A,\tilde B)$ be the characteristic representative of the  $\mathfrak{g}$-valued ZCR $(\restr{A}{\mathcal{E}},\restr{B}{\mathcal{E}})$. Then the equality 
\begin{equation}\label{ZCRcondition}\restr{\sum_{I}(-1)^{|I|}{\widehat D}_I\frac{\partial}{\partial u_I^k}\left(D_y\widetilde A-D_x\widetilde B+[\widetilde A,\widetilde B]\right)}{\mathcal{E}}=0 \end{equation}
holds for all $ k=1,\dots,m.$
\end{proposition}

\begin{proof}
Let $Q=(Q_1,\dots,Q_N)$ be the characteristic corresponding to the characteristic representative $(\tilde A,\tilde B)$ .  We have
\begin{align*}
0&=\restr{\sum_{I,l}(-1)^{|I|}{\widehat D}_I\left(\frac{\partial F^l}{\partial u_I^k}Q_l\right)}{\mathcal{E}}=\restr{\sum_{I,l}(-1)^{|I|}{\widehat D}_I\left(\frac{\partial F^l}{\partial u_I^k}Q_l\right)}{\mathcal{E}}+\restr{\sum_{I,l}(-1)^{|I|}{\widehat D}_I\left(F^l\frac{\partial Q_l}{\partial u_I^k}\right)}{\mathcal{E}}\\
&=\restr{\sum_{I}(-1)^{|I|}{\widehat D}_I\left(\frac{\partial}{\partial u_I^k}\left(\sum_{l}F^lQ_l\right)\right)}{\mathcal{E}}=\restr{\sum_{I}(-1)^{|I|}{\widehat D}_I\left(\frac{\partial}{\partial u_I^k}\left(D_y\widetilde A-D_x\widetilde B+[\widetilde A,\widetilde B]\right)\right)}{\mathcal{E}}.
\end{align*}
The first equality follows from   \eqref{Nec}, since $\restr{Q}{\mathcal{E}}$ is the characteristic element of $(\restr{A}{\mathcal{E}},\restr{B}{\mathcal{E}})$. The second equality holds since the term $\sum_lF^l\frac{\partial Q_l}{\partial u_I^k}$ lies in the ideal $\mathcal{I}^{\mathfrak{g}}(\mathcal{E})$, that is, $\restr{\sum_lF^l\frac{\partial Q_l}{\partial u_I^k}}{\mathcal{E}}=0$, and both $\widehat D_x^{\widetilde A}$ and $\widehat D_y^{\widetilde B}$ are linear on $\mathcal{F}^{\mathfrak{g}}(\mathcal{E})$, thus they map zero to zero. The other two equalities follow from the linearity of the operators $\widehat D_x^{\widetilde A}$, $\widehat D_y^{\widetilde B}$ and $\frac{\partial}{\partial u_I^k}$ on $\mathcal{F}^{\mathfrak{g}}(\pi)$.
\end{proof}

Proposition~\ref{newNec}, together with Proposition~\ref{CharForm}, yields the result formulated in the following corollary.
\begin{corollary}\label{cor2}
Let $\restr{(A,B)}{\mathcal{E}}$ be a $\mathfrak{g}$-valued ZCR for $\mathcal{E}$. There exists its representative $(\widetilde A,\widetilde B)$ such that 
\begin{equation*}\restr{\sum_{I}(-1)^{|I|}{\widehat D}_I\frac{\partial}{\partial u_I^k}\left(D_y\widetilde A-D_x\widetilde B+[\widetilde A,\widetilde B]\right)}{\mathcal{E}}=0\end{equation*}
for all $k=1,\dots,m$.
\end{corollary}
\begin{remark}\label{triv}
Let us note that in the case where $\mathfrak{g}$ is abelian, the condition~\eqref{ZCRcondition} just proved is satisfied trivially by arbitrary $\mathfrak{g}$-valued functions $A$ and $B\in\mathcal{F}^{\mathfrak{g}}(\pi)$ (thus, the condition  is trivially satisfied also in the case where $(-B, A)$ is a conserved current).\\
Indeed, in the abelian case, the operator 
$
\sum_{I}(-1)^{|I|}{\widehat D}_I \circ \frac{\partial}{\partial u^k_I} : \mathcal{F}^{\mathfrak{g}}(\pi) \to \mathcal{F}^{\mathfrak{g}}(\pi)
$
coincides with the $k$-th component $\mathrm{E}_{k}^{\mathfrak{g}}$ of the Euler operator $\mathrm{\mathbf{E}}^{\mathfrak{g}}$, and for arbitrary $\mathfrak{g}$-valued functions $A,B\in\mathcal{F}^{\mathfrak{g}}(\pi)$ we have
$
D_y A - D_x B + [A, B] = \mathrm{Div}^{\mathfrak{g}}(-B, A),
$
hence,
\begin{equation*}
\sum_{I}(-1)^{|I|}{\widehat D}_I \left( \frac{\partial}{\partial u_I^k} \left( D_y A - D_x B + [A, B] \right) \right) 
= \mathrm{E}^{\mathfrak{g}}_k(\mathrm{Div}^{\mathfrak{g}}(-B, A)) = 0 \quad \text{identically on } J^{\infty}(\pi).
\end{equation*}
Thus, we have also
\begin{equation*}
\restr{\sum_{I}(-1)^{|I|}{\widehat D}_I \left( \frac{\partial}{\partial u_I^k} \left( D_y A - D_x B + [A, B] \right) \right)}{\mathcal{E}} = 0.
\end{equation*}
 for arbitrary $A,B\in\mathcal{F}^{\mathfrak{g}}(\pi)$.
\end{remark}

Note, that similarly to the abelian case, the left-hand side of condition~\eqref{ZCRcondition} can be interpreted as the result of applying a gauge-theoretic analogue of the Euler operator to the $\mathfrak{g}$-valued function $D_y A - D_x B + [A, B]$ on $J^{\infty}(\pi)$, however, followed by restriction to $\mathcal{E}$.

More precisely, given a representative $(A,B)$ of a ZCR $\restr{(A,B)}{\mathcal{E}}$ for $\mathcal{E}$, let us introduce the following notation:
\begin{equation*}\widehat{\mathbf{E}}^{A,B}= \left( \widehat{\mathrm{E}}_1^{A,B}, \dots, \widehat{\mathrm{E}}_N^{A,B} \right):\mathcal{F}^{\mathfrak{g}}(\pi)\to(\mathcal{F}^{\mathfrak{g}}(\pi))^m,\end{equation*}
\begin{equation*}
\widehat{\mathrm{E}}_k^{A,B} := \sum_I (-1)^{|I|} \widehat D_I \circ \frac{\partial}{\partial u_I^k},
\end{equation*}
where the sum is taken over all multi-indices $I = (a, b)$ with $a, b \in \mathbb{N} \cup \{0\}$, and with corresponding operators $\widehat D_I = (\widehat D_x^A)^a \circ (\widehat D_y^B)^b$. The order of composition is specified (first $x$, then $y$) because, in general, $\widehat D_x^A$ and $\widehat D_y^B$ do not commute on $\mathcal{F}^{\mathfrak{g}}(\pi)$. Nevertheless, this choice is without loss of generality, since the result obtained after restriction to $\mathcal{E}$ remains unchanged, when the operators $\widehat D_x^A$ and $\widehat D_y^B$ are applied in different order ($\widehat D_x^A$ and $\widehat D_y^B$ commute on $\mathcal{F}^{\mathfrak{g}}(\mathcal{E})$ ).

Thus, ~\eqref{ZCRcondition} can be written as
\begin{equation*}
\restr{\left(\widehat{\mathrm{E}}_k^{A,B} (D_y A - D_x B + [A, B])\right)}{\mathcal{E}}=0 \mbox{ for all }k=1,\dots,m.
\end{equation*}

In the following example, we illustrate the validity of identity~\eqref{ZCRcondition} for characteristic representatives of a ZCR using concrete examples. It turns out that, unlike in the abelian case, the expression $\widehat{\mathrm{E}}_k^{A,B}(D_y A - D_x B + [A, B])$ associated with a characteristic representative $(A, B)$ of a ZCR does not necessarily vanish identically on the entire jet space $J^{\infty}(\pi)$. Moreover, we show that condition~\eqref{ZCRcondition} is, in general, not satisfied for non-characteristic representatives of ZCRs. This, \textit{inter alia},  shows that condition ~\eqref{ZCRcondition} is by no means satisfied trivially, nor it is a direct consequence of the Maurer--Cartan condition \eqref{ZCR}.

\begin{example}\label{new1}
Consider again the equation manifold
\begin{equation*}
\mathcal{E}:\quad F = u_y - u_{xxx} + 6u u_x = 0
\end{equation*}
given by the KdV equation along with its differential consequences, and consider the ZCR $\restr{(A,B)}{\mathcal{E}}$ for $\mathcal{E}$ from Example~\ref{ex1}. Let us successively examine all three representatives $(\bar A,\bar B)$, $(\widetilde A,\widetilde B)$, and $(A,B)$ of this ZCR introduced in Example~\ref{ex1}, and apply to each the relevant gauge-theoretic analogue of the Euler operator, as defined above.

In the case of $(\bar A,\bar B)$ we have
\begin{align*}
\widehat{\mathrm{E}}^{\bar A,\bar B}(D_y\bar A-D_x\bar B+[\bar A,\bar B])
&=\widehat{\mathrm{E}}^{\bar A,\bar B}\left(\begin{pmatrix}0&0\\6uu_x-u_{xxx}+u_y&0\end{pmatrix}\right)\\
&=\frac{\partial}{\partial u}\begin{pmatrix}0&0\\6uu_x-u_{xxx}+u_y&0\end{pmatrix}
-\widehat D_x\left(\frac{\partial}{\partial u_x}\begin{pmatrix}0&0\\6uu_x-u_{xxx}+u_y&0\end{pmatrix}\right)\\
& -\widehat D_y\left(\frac{\partial}{\partial u_y}\begin{pmatrix}0&0\\6uu_x-u_{xxx}+u_y&0\end{pmatrix}\right)
-\widehat D_x^3\left(\frac{\partial}{\partial u_{xxx}}\begin{pmatrix}0&0\\6uu_x-u_{xxx}+u_y&0\end{pmatrix}\right)\\
&=\begin{pmatrix}0&0\\6u_x&0\end{pmatrix}
-D_x\begin{pmatrix}0&0\\6u&0\end{pmatrix}
+\mathrm{ad}_{\bar A}\begin{pmatrix}0&0\\6u&0\end{pmatrix}
-D_y\begin{pmatrix}0&0\\1&0\end{pmatrix}
+\mathrm{ad}_{\bar B} \begin{pmatrix}0&0\\1&0\end{pmatrix}\\
& -\left(D_x-\mathrm{ad}_{\bar A}\right)^3\begin{pmatrix*}[r]0&0\\-1&0\end{pmatrix*}
\end{align*}
However, direct computations show that 
\begin{equation*}
\mathrm{ad}_{\bar A}\begin{pmatrix}0&0\\6u&0\end{pmatrix}=\begin{pmatrix}6u&0\\0&-6u\end{pmatrix},\ 
\mathrm{ad}_{\bar B} \begin{pmatrix}0&0\\1&0\end{pmatrix}=\begin{pmatrix}-2u&0\\-2u_x&2u\end{pmatrix}
\end{equation*}
and,
\begin{equation*}
\left(D_x-\mathrm{ad}_{\bar A}\right)^3\begin{pmatrix*}[r]0&0\\-1&0\end{pmatrix*}=\begin{pmatrix*}[r]4u&0\\-2u_x&-4u
\end{pmatrix*}.\end{equation*}
We thus obtain an $\mathfrak{sl}(2)$-valued function that is identically zero on $J^{\infty}(\pi)$:
\begin{equation*}\widehat{\mathrm{E}}^{\bar A,\bar B}(D_y\bar A-D_x\bar B+[\bar A,\bar B])=\begin{pmatrix*}[r]0&0\\0&0\end{pmatrix*}.\end{equation*}
The $\mathfrak{sl}(2)$-valued function $\widehat{\mathrm{E}}^{\widetilde A,\widetilde B}(D_y\widetilde A-D_x\widetilde B+[\widetilde A,\widetilde B])$ can be computed analogously, however, the computation is significantly more involved. For this reason, we used computer algebra software to carry it out. We do not present the full result here, as it is rather lengthy. What is important, though, is the fact that
\begin{equation*}
\widehat{\mathrm{E}}^{\widetilde A,\widetilde B}(D_y\widetilde A - D_x\widetilde B + [\widetilde A,\widetilde B]) = X,
\end{equation*}
where $X$ is a $\mathfrak{g}$-valued function on $J^{\infty}(\pi)$ that belongs to the ideal $\mathcal{I}^{\mathfrak{g}}(\mathcal{E})$, but is not identically zero on $J^{\infty}(\pi)$.

We thus observe that in both cases of characteristic representatives of the ZCR for the KdV equation, the application of the corresponding gauge-theoretic analogue of the Euler operator to the left-hand side of the Maurer-Cartan condition always yields a $\mathfrak{sl}(2)$-valued function belonging to the ideal $\mathcal{I}^{\mathfrak{sl}(2)}(\mathcal{E})$. Therefore, condition~\eqref{ZCRcondition} is indeed satisfied, precisely as stated in Proposition~\ref{newNec}. However, this $\mathfrak{g}$-valued function may or may not vanish identically on the entire space $J^{\infty}(\pi)$.

Finally, let us carry out an analogous computation for the non-characteristic representative $(A,B)$.
Upon the use of the computer algebra software we obtain 
\begin{equation*}
\widehat{\mathrm{E}}^{A,B}(D_y A-D_x B+[A,B])
=\begin{pmatrix*}[r]f_{1}&f_2\\f_3&-f_1\end{pmatrix*}+Y,
\end{equation*}
where
\begin{equation*}
f_1 = 288u^2u_x + 20uu_{xxx} +124u_xu_{xx} -8u_{5x},\ 
f_2 = 64u^3 + 152uu_{xx} + 160u_x^2 - 20u_{4x}\end{equation*}
\begin{equation*}f_3 = -64u^4 - 168u^2u_{xx} - 280uu_x^2 + 16uu_{4x} - 16u_xu_{xxx} - 20u_{xx}^2 + 2u_{6x},\end{equation*}
and $Y$ is a $\mathfrak{g}$-valued function on $J^{\infty}(\pi)$ whose explicit form is rather lengthy and thus omitted here, but it can be verified that it belongs to the ideal $\mathcal{I}^{\mathfrak{sl}(2)}(\mathcal{E})$.

On the other hand, the matrix
$
\begin{pmatrix*}[r]f_{1}&f_2\\f_3&-f_1\end{pmatrix*}
$
evidently does not belong to $\mathcal{I}^{\mathfrak{sl}(2)}(\mathcal{E})$. Thus, the condition~\eqref{ZCRcondition} is not satisfied for $(A,B)$. It is therefore clear that the statement of Proposition~\ref{newNec} does not hold for an arbitrary representative of a ZCR.
\end{example}

Up to this point, it might seem that characteristic representatives can be viewed as a kind of normal form for all  $\mathfrak{g}$-valued zero-curvature representations of a given equation $\mathcal{E}$ -- even though the characteristic form is not unique.
However, zero-curvature representations are typically considered only up to gauge equivalence.
For our characteristic form to be regarded as a proper normal form, it must be preserved under gauge transformations.
It is generally known that this is indeed the case.
For the sake of completeness, we formulate this result explicitly as a proposition below. Based on the discussion in Remark~\ref{gaugeCL1}, the following proposition can be regarded as  a natural extension of Classical Result~\ref{prop_CL_char} which concerns conservation laws to the broader setting of $\mathfrak{g}$-valued zero-curvature representations.
 \begin{proposition}\label{prop_gauge}
Let $(A,B)$ be the characteristic representative of a $\mathfrak{g}$-valued ZCR for $\mathcal{E}$ with associated characteristic $Q=(Q_1,\dots,Q_N)$ and let $H\in\mathcal{F}^{\mathcal{G}}(\pi)$ be a $\mathcal{G}$-valued function. Then the transformed pair $(A^H,B^H)$ is a characteristic representative of the gauge-transformed ZCR $\restr{(A^H,B^H)}{\mathcal{E}}$, and the characteristic associated to $(A^H,B^H)$ is given by 
\begin{equation*}Q^H=(HQ_1H^{-1},\dots,HQ_NH^{-1}).\end{equation*}
\end{proposition}\label{gauge}
\begin{proof}
By a straightforward computation, and using the identities \eqref{der-inv} we obtain

\begin{align}
D_y A^H - D_x B^H + [A^H, B^H] 
&= H \left( D_y A - D_x B + [A, B] \right) H^{-1} 
= H \left( \sum_{l} F^lQ_l\right) H^{-1}\nonumber \\
&= \sum_{l}F^l \cdot (H Q_lH^{-1})\nonumber,
\end{align}
where for each term $ H Q_lH^{-1}$ we have $ H Q_lH^{-1}\in \mathcal{F}^{\mathfrak{g}}(\pi)$, as the Lie algebra $\mathcal{F}^{\mathfrak{g}}(\pi) $ is closed under conjugation by elements of $ \mathcal{F}^\mathcal{G}(\pi)$.

\end{proof}

Thus, if one wishes to search for, up to gauge equivalence, all $\mathfrak{g}$-valued ZCRs for a given equation, it suffices to search for their characteristic representatives modulo gauge transformations. To perform the classification modulo gauge transformations,  it is possible to rely on ideas and results concerning the normal forms of ZCRs obtained in \cite{Ma97},\cite{Se05} and \cite{Se08}, slightly modified for $N$-tuples of $\mathfrak{g}$-valued functions on the entire $J^{\infty}(\pi)$.

\section{Concluding remarks}

The main aim of this paper was to examine $\mathfrak{g}$-valued zero-curvature representations and their characteristic elements for equations in two independent variables as a natural extension of the concept of conservation laws and their generating functions -- primarily from the perspective of their extension to the entire jet space. Our goal was to generalize certain classical results known for conservation laws -- namely Classical Result 1 through 4 as formulated in Section \ref{sec:CL} -- to the nonabelian setting. 

We have succeeded in generalizing three of the classical results, and these generalizations hold  for all $\mathfrak{g}$-valued zero-curvature representations, regardless of any further properties of the Lie algebra $\mathfrak{g}\subset\mathfrak{gl}(n)$.
\begin{itemize}
\item A generalization of Classical Result 1 was achieved when we proved that for every $\mathfrak{g}$-valued zero-curvature representation, the corresponding Maurer--Cartan condition can be rewritten in a form analogous to the characteristic form of a conservation law (Proposition \ref{CharForm}).
\item Classical Result 2 was extended to the case of ZCRs once we recognized that the right-hand side of the Maurer--Cartan condition, expressed in characteristic form, can be interpreted as the `scalar product' of the functions defining the given equation with one of the many possible extensions of the characteristic element of the ZCR under consideration (Corollary \ref{cor}).
\item A generalization of Classical Result 4 to the case of $\mathfrak{g}$-valued ZCRs follows from the observation that gauge equivalence of $\mathbb{R}$-valued ZCRs coincides with the (second-kind) equivalence of conservation laws (Remark \ref{gaugeCL1},\ Proposition \ref{prop_gauge}).
\end{itemize}

In contrast, we were able to fully extend \textbf{Classical Result~3} to $\mathfrak{g}$-valued ZCRs only in the case where $\mathfrak{g}$ is an abelian matrix Lie algebra. In the nonabelian setting, the situation is more delicate. Nevertheless, we have succeeded in formulating a new, nontrivial necessary condition on ZCRs formulated in terms of their characteristic representatives (Corollary \ref{cor2}), which is trivially satisfied in the abelian case, as it reduces to the restriction of identity~\eqref{ZCRprop4char_suf} to~$\mathcal{E}$. In this sense, it may be viewed as a step toward extending Classical Result~3 to the nonabelian case.

Since the left-hand side of condition~\eqref{ZCRcondition} arises from applying a gauge-analogue of the Euler operator -- one of many possible candidates, due to the noncommutativity of the operators $\widehat{D}_x^A$ and $\widehat{D}_y^B$ on $\mathcal{F}^{\mathfrak{g}}(\pi)$ and due to the many possible representatives of $\restr{A}{\mathcal{E}}$ and $\restr{B}{\mathcal{E}}$ -- to the left-hand side of the characteristic form of the Maurer--Cartan condition and then restricting to~$\mathcal{E}$, it is natural to ask whether there exists a distinguished operator (uniquely determined by the given ZCR and the defining functions of~$\mathcal{E}$) whose application yields an identically vanishing expression on the entire jet space~$J^\infty(\pi)$.
The existence of such an operator would represent a genuine extension of the ``only if'' part of \textbf{Classical Result~3} to the nonabelian setting.

In Example~\ref{new1}, we successfully identified such an operator for the characteristic representative $(\bar A, \bar B)$ of a ZCR for the KdV equation  -- it was the operator $\widehat{\mathrm{E}}^{\bar A,\bar B}$.
However, for the characteristic representative $(\widetilde{A}, \widetilde{B})$ -- which corresponds to the same ZCR -- we were unable to find such an operator. \textit{This question therefore remains open.}

Although the results obtained in this paper are primarily of theoretical nature, we believe they could find further applications.

The result formulated in Proposition~\ref{CharForm}, together with Proposition~\ref{gauge}, shows that the characteristic form of a ZCR can be regarded as a kind of normal form for all $\mathfrak{g}$-valued ZCRs of a given (general) equation. As such, it can be used for classification purposes---for example, for a fixed Lie algebra $\mathfrak{g}$ and a given gauge class of function pairs represented by $(A,B)$ together with a corresponding $N$-tuple $(Q_1,\dots,Q_N)$, one may use it to ask the question which equations admit $(A,B)$ as a characteristic representative of a ZCR with characteristic $(Q_1,\dots,Q_N)$.

Moreover, the results obtained in this paper could be applied in the computational search for unknown $\mathfrak{g}$-valued ZCRs associated with a given equation. In the works~\cite{Ma92,Ma97}, the characteristic element and its transformation behaviour under gauge-transformations play a central role in such computations. However, those approaches rely solely on the necessary condition~\eqref{Nec} imposed on the characteristic element, together with the Maurer--Cartan condition~\eqref{ZCR}. We believe that the sufficient condition formulated in Corollary~\ref{cor} could enhance the effectiveness of these procedures. In analogy with the methods used to compute conservation laws described in Table 5 of~\cite{Wolf}, the existing approaches to computing ZCRs correspond to steps IA and IB. In contrast, our result would allow for the combined use of steps IB and IIA, as recommended by the author in the case of conservation laws.

Finally, the result formulated in Corollary \ref{cor2} provides, in addition to the Maurer--Cartan condition, $m$ further conditions expressed purely in terms of the characteristic representative. These may also be useful in computation.

\section*{Acknowledgements}
The author thanks M.~Marvan for stimulating discussions.  
The symbolic computations were performed in \textsc{Maple} with the help of the open-source package \textsc{Jets} \cite{Jets}.\\
This work was supported by the Ministry of Education, Youth and Sports of the Czech Republic (RVO IC47813059).\\


\begin{thebibliography}{99}
\bibitem{Bal15}  A.~V.~Balandin, \emph{Characteristics of Conservation Laws of Chiral-Type Systems}, Lett. Math. Phys. {105} (2015),  27-43. \href{https://doi.org/10.1007/s11005-014-0736-8}{DOI:10.1007/s11005-014-0736-8}
\bibitem{Bal16} A.~V.~Balandin, \emph{Tensor fields defined by Lax representations}, J. Nonlinear Math. Phys. \textbf{23} (2016), no.~3, 323--334. \href{https://doi.org/10.1080/14029251.2016.1199494}{DOI:10.1080/14029251.2016.1199494}
\bibitem{Jets} H.~Baran, M.~Marvan, \emph{Jets. A software for differential calculus on jet spaces and diffieties}.  \url{http://jets.math.slu.cz}.
\bibitem{Boch} A.V. Bocharov et al., \emph{Symmetries and Conservation Laws for Differential Equations of Mathematical Physics}, AMS, Providence, RI, 1999.
\bibitem{Igonin2019} S. Igonin, G. Manno, \emph{Lie algebras responsible for zero-curvature representations of scalar evolution equations}, J. Geom. Phys. \textbf{138} (2019), 297--316. \href{https://doi.org/10.1016/j.geomphys.2018.10.019}{DOI:10.1016/j.geomphys.2018.10.019}
\bibitem{IgoninManno2020} S. Igonin and G. Manno, \emph{On Lie algebras responsible for integrability of (1+1)-dimensional scalar evolution PDEs}, Journal of Geometry and Physics {\bf 150} (2020), 103596. \href{https://doi.org/10.1016/j.geomphys.2020.103596}{DOI:10.1016/j.geomphys.2020.103596}
\bibitem{Kra11} I.S. Krasil{\cprime}shchik, A. M. Verbovetsky,  \emph{Geometry of jet spaces and integrable systems}, J. Geom. Phys. {\bf 61} (2011), no. 9, 1633--1674. \href{https://doi.org/10.1016/j.geomphys.2010.10.012}{DOI:10.1016/j.geomphys.2010.10.012}, \href{https://arxiv.org/abs/1002.0077}{arXiv:1002.0077v6}
\bibitem{Kra17} I.S. Krasil{\cprime}shchik, A.M. Verbovetsky, R. Vitolo, \emph{The Symbolic Computation of Integrability Structures for Partial Differential Equations}, Texts \& Monographs in Symbolic Computation, Springer, Berlin (2017). \href{https://doi.org/10.1007/978-3-319-71655-8}{DOI:10.1007/978-3-319-71655-8}.
\bibitem{Ma92} M. Marvan, \emph{On zero-curvature-representations of partial differential equations},\  in: Differential Geometry and Its Applications, Proc. Conf. Opava, Czechoslovakia, Aug. 24?28, 1992, Math. Publ. 1 (Silesian University, Opava, 1993) 103--122.
\bibitem{Ma97} M. Marvan, \emph{A direct procedure to compute zero-curvature representations. The case sl2}, Secondary Calculus and Cohomological Physics, Proc. Conf. Moscow 1997, pp.10, 1998.
\bibitem{Ma02} M. Marvan, \textit{Scalar second-order evolution equations possessing an irreducible $sl_2$-valued zero-curvature representation}, J. Phys. A: Math. Gen 35 (2002), 9431-9439. \href{https://iopscience.iop.org/article/10.1088/0305-4470/35/44/312}{10.1088/0305-4470/35/44/312}
\bibitem{Ma02-2} M. Marvan, \textit{On the Horizontal Gauge Cohomology and Nonremovability of the Spectral Parameter}, Acta Applicandae Mathematicae 72 (2002), 51--65, \href{https://doi.org/10.1023/A:1015218422059}{DOI:10.1023/A:1015218422059}
\bibitem{Ol93} P. J. Olver, \textit{Applications of Lie Groups to Differential Equations}, 2nd edition, Springer-Verlag, 1993.
\bibitem{Sa95} S. Yu. Sakovich,\textit{On zero-curvature representations of evolution equations},\ J. Phys. A: Math. Gen 28 (1995),\ 2861--2869. \href{https://doi.org/10.1088/0305-4470/28/10/016}{DOI:10.1088/0305-4470/28/10/016}
\bibitem{Sak04} S. Yu. Sakovich, \emph{Cyclic Bases of Zero-Curvature Representations: Five Illustrations to One Concept}, Acta Appl. Math. {\bf 83} (2004), 69-83. \href{https://doi.org/10.1023/B:ACAP.0000035589.61486.a7}{DOI:10.1023/B:ACAP.0000035589.61486.a7}

\bibitem{Sa11} S.~Yu. Sakovich, \emph{Integrability of the Bakirov System: A Zero-Curvature Representation}, 
International Journal of Mathematics and Mathematical Sciences, {\bf 2011} (2011), 1--7.
\href{https://doi.org/10.1155/2011/497828}{DOI:10.1155/2011/497828}
\bibitem{Se05}P. Sebesty\'{e}n, \emph{Normal forms of irreducible $\mathfrak{sl}(3)$-valued zero-curvature representations}, Reports on Mathematical Physics 55 (2005), no. 3, 435-445. \href{https://doi.org/10.1016/S0034-4877(05)80057-6}{DOI:10.1016/S0034-4877(05)80057-6}

\bibitem{Se08}
P. Sebesty\'{e}n, \emph{On normal forms of irreducible $\mathfrak{sl}(n)$-valued zero-curvature representations}, Reports on Mathematical Physics 62 (2008), no. 1, 57-68. \href{https://doi.org/10.1016/S0034-4877(08)00021-9}{DOI:10.1016/S0034-4877(08)00021-9}

\bibitem{Wahl75} H.D. Wahlquist, F. B. Estabrook. \emph{Prolongation structures of nonlinear evolution equations} I,\ II, J. Math. Phys. 16 (1975), 1--7. \href{https://doi.org/10.1063/1.522396}{DOI:10.1063/1.522396}
\bibitem{Wolf}
T. Wolf, \emph{A comparison of four approaches to the calculation of conservation laws}, European Journal of Applied Mathematics {\bf 13} (2002), no. 2, 129--152. \href{https://doi.org/10.1017/S0956792501004715}{DOI:10.1017/S0956792501004715}

\bibitem{Zak} V. E. Zakharov,\ A. B. Shabat,\ \textit{Integrovanie nelinejnykh uravnenij matematicheskoj fiziki metodom obratnoj zadachi rasseyaniya, } Funkc. Anal. Prilozh. 13 (1979) (3), 13--22.
\end{thebibliography}
\end{document}